%% file: report.tex
\documentclass[12pt]{article}
\input{preamble.ltx}
\usepackage{rotating}

\newif\ifanonymous
\anonymousfalse  

\numberwithin{equation}{section}

\ifanonymous
\else
\fi

\title{\sffamily Subsample-Based Estimation under \\ Dynamic Contamination}

\ifanonymous
\date{Submitted for review}
\else
\author{
\sffamily Yukai Yang\orcidlink{0000-0002-2623-8549}\thanks{Corresponding author: \href{mailto:yukai.yang@statistik.uu.se}{yukai.yang@statistik.uu.se}}
\\
\small Department of Statistics, Uppsala University
\and
\sffamily Rickard Sandberg\orcidlink{0000-0003-0589-4034}
\\
\small Department of Entrepreneurship, Innovation and Technology, \\
\small Center for Data Analytics, Stockholm School of Economics
}
\date{}
\fi

\begin{document}

\pagenumbering{arabic}

\maketitle

\begin{abstract}
This paper studies a structural failure of subsample-based estimation in dynamic time series models. Even under oracle knowledge of contamination locations, removing contaminated observations does not restore the uncontaminated objective. In such settings, contamination propagates through the residual filter and distorts the estimation criterion. As a result, subsample-based estimators are generically inconsistent for the clean-data parameter. We characterise this failure as a structural incompatibility between pointwise subsampling and residual propagation. More generally, the failure arises whenever contamination propagates through transformations that enter the estimation criterion, with dynamic time series models as a leading example. To address it, we propose a propagation-compatible transformation of index sets via a patch removal operator. Under general high-level conditions, this transformation leaves the estimator asymptotically unchanged under the uncontaminated model while restoring consistency under contamination. The results apply to a broad class of residual-based estimators and do not rely on modelling the contamination process.
\end{abstract}




\section{Introduction}\label{sec:introduction}

Outliers are a pervasive feature of economic and financial time series, and more generally of data in many statistical applications.
Their primary statistical effect is not merely to generate a few anomalous observations, but to distort inference about the underlying data generating mechanism.
Such contamination may induce spurious statistical significance \citep{Huber1964, Tsay1986, ChenLiu1993, DaviesGather1993} and create false evidence of structural features such as volatility clustering \citep{FransesGhijsels1999, CarneroPenaRuiz2007}, structural breaks, regime switching \citep{Andrews1993, vDFL99}, or persistence and unit-root-like behaviour, even when the clean process exhibits no such features.
From an econometric perspective, the object of interest is typically the latent clean process rather than the contaminated observations themselves.
Outliers act as external perturbations that may mimic or obscure level shifts and variance changes \citep{ChenLiu1993, BoxTiao1975, Tsay1988, FransesHaldrup1994}. 
We refer to such properties of the underlying data-generating mechanism as structural features.
If contamination causes the observed series to appear to exhibit a structural feature that is absent from the clean process, inference based on the observed data is misleading.
Reliable inference therefore requires consistent estimation of the clean model under contamination, since asymptotic approximations, hypothesis testing, model selection, and specification searches all depend on an accurate characterisation of the underlying dynamics \citep{JohansenNielsen2009, DoornikHendry2016}.

A common and widely used approach in this setting is to construct estimators based on subsets of the data intended to exclude contaminated observations. 
Such subsample-based procedures operate without explicit modelling of the contamination mechanism and avoid imposing assumptions on the contamination process itself. 
They form a standard paradigm for robust estimation when the objective is to recover the clean-data mechanism \citep{Huber1964, Ro84, Maronna2006, JohansenNielsen2016}. 
This objective differs fundamentally from approaches that incorporate contamination directly into the data-generating mechanism, where contamination is treated as part of the model and therefore requires additional distributional and structural assumptions for specification and identification.
The present paper adopts the former perspective and treats contamination as an external nuisance rather than as part of the data-generating mechanism. 
Their use in time series, however, implicitly rests on the idea that removing contaminated observations restores the uncontaminated optimisation problem.
As we demonstrate below, contamination propagates through the residual filter that defines the estimation criterion, inducing non-local distortions in the residual sequence and hence in the criterion itself.

We show that subsample-based procedures do not generically provide consistent estimation of the clean-data mechanism in dynamic time series models, even under oracle knowledge of contamination locations. 
In cross-sectional settings, contamination is local, so removing contaminated observations restores the clean-data objective. 
In dynamic models, however, this equivalence fails. 
As a result, robustness cannot be achieved by pointwise removal of contaminated observations, but requires control of propagated distortions in the residual-based criterion.

We establish a general inconsistency result for subsample-based procedures under dynamic contamination, showing that this failure is structural and arises from residual propagation through the filtering mechanism rather than from imperfect identification of contaminated observations. 
We characterise this mechanism by deriving an exact representation of additive and innovative contamination after passage through the residual filter, which induces non-local distortions in the residual-based criterion. 
We show that this implies a propagation-compatible transformation of index sets that removes the residual footprint of contamination, yielding a patch removal operator governed by the model-implied filter rather than by an ad hoc trimming rule and applicable to a broad class of residual-based procedures without modifying their internal structure. 
We establish asymptotic invariance under the transformation in the uncontaminated case and asymptotic equivalence between contaminated and uncontaminated estimators once residual propagation is controlled, thereby restoring consistency for the clean-data parameter under contamination.

Existing work on outliers in time series has focused on their effects on residuals, diagnostics, and detection procedures.
Early work, beginning with \citet{fox_outliers_1972} and later surveyed by \citet{MARTIN1979147} and \citet{MaYo86}, studies additive and innovative outliers and their effects on residuals, diagnostics, and influence functions in dependent data.
Classical model-based approaches include \citet{chang_estimation_1988} and \citet{ChenLiu1993}, who analyse outlier effects and develop procedures for parameter estimation and joint estimation of model and outlier effects in ARMA-type settings.
More recent contributions, particularly \citet{JohansenNielsen2016}, develop a unified asymptotic theory for outlier detection procedures, formalising their behaviour through the notion of \emph{gauge} and establishing calibration results for the fraction of observations flagged as outliers. 

These detection-based procedures are closely related to residual-based screening, down-weighting, and trimming methods, including M-type estimators \citep{Huber1964, WelshRonchetti2002}, least trimmed squares (LTS) and related high-breakdown estimators \citep{Ro84, Visek2006a}, and forward search methods \citep{HadiSimonoff1993, AtkinsonRiani2000}. 
While this provides a rigorous basis for detection and tuning, these procedures operate at the level of individual observations or residual magnitudes. 
They proceed through identification and removal, down-weighting, or trimming of aberrant observations based on residual diagnostics, a well-known difficulty in time series settings \citep{Ljung93}. 
Correct identification or asymptotic calibration of contaminated observations is not sufficient to restore the uncontaminated optimisation problem. 
The difficulty is structural rather than algorithmic. Even when contaminated observations are correctly excluded, residual propagation continues to distort the residual-based criterion through the dynamic filter. 
The resulting optimisation problem therefore no longer targets the clean-data parameter asymptotically.
This limitation also applies to robust estimation for multivariate dynamic models, such as VARMA settings \citep{GarciaBenMartinezYohai1999}, and to approaches based on modified residual constructions that limit the effect of individual outliers \citep{MulerPenaYohai2009}. 
Despite decades of work on outlier detection and robust estimation in time series, 
the question of whether consistency for the clean-data parameter can be restored 
under dynamic contamination without modelling the contamination process has remained open. 
The central issue is therefore not whether contaminated observations can be identified, but whether removing them restores the uncontaminated estimation problem.

\citet{JohansenNielsen2016} emphasise clustered contamination and $\varepsilon$-contamination as important directions for future research. 
We show that even sparse contamination, when combined with temporal dependence, can undermine consistency through residual propagation, and we provide a general correction that restores consistency for the clean-data parameter under high-level conditions that apply broadly across residual-based estimators.

The remainder of the paper is organised as follows. 
Section~\ref{sec:model} introduces the contaminated VARMA framework and characterises residual propagation. 
Section~\ref{sec:estimation} develops the general estimation framework and formalises the propagation-compatible transformation of index sets. 
Section~\ref{sec:verification} discusses the verification of the assumptions and establishes sufficient conditions under which they hold in standard settings.
Section~\ref{sec:simulation} reports simulation evidence. 
Section~\ref{sec:algorithm} presents a feasible implementation. 
Section~\ref{sec:illustration} provides an empirical illustration. 
Section~\ref{sec:conclusion} concludes.

\section{Contamination and Its Residual Propagation} \label{sec:model}

Let $\{x_t\}_{t\in\mathbb Z}$ be a $d$-dimensional stochastic process generated by a vector autoregressive moving average (VARMA) model
\begin{equation} \label{eqn:arma}
\phi_0(L) x_t = \theta_0(L) \varepsilon_t,
\end{equation}
where $\phi_0(L)$ and $\theta_0(L)$ are lag polynomials with orders $p$ and $q$ and $\varepsilon_t$ is a $d$-dimensional zero mean innovation vector.

In this paper, we follow Definitions 3.1.3 and 3.1.4 in \citet{brockwell_time_1991}. 
The VARMA model \eqref{eqn:arma} is said to be causal if and only if it admits the representation
\begin{equation} \label{eqn:inftymv}
x_t = \psi_0(L) \varepsilon_t,
\end{equation}
where $\psi_0(L) = \phi_0^{-1}(L)\theta_0(L)$.
It is said to be invertible if and only if it admits the representation
\begin{equation} \label{eqn:inftyar}
\pi_0(L) x_t = \varepsilon_t,
\end{equation}
where $\pi_0(L) = \theta_0^{-1}(L)\phi_0(L)$.

Let $\mathcal{V}$ denote the parameter space and let $\varphi_0=(\phi_0,\theta_0)\in\mathcal{V}$ denote the true parameter vector of the data generating process.
We restrict the parameter space to values for which the VARMA model is causal and invertible. 
Under these conditions, the associated filters $\psi(L)$ and $\pi(L)$ exist and are absolutely summable.

To focus on the stochastic dynamic structure, we consider a mean-zero VARMA process without deterministic components.
In practice, deterministic terms such as intercepts, trends, seasonal components, or exogenous regressors can be removed by standard preprocessing and are therefore not included in the present formulation.

Outliers are commonly modelled as either additive outliers (AO) or innovative outliers (IO) \citep{fox_outliers_1972, AbrahamBox1979, chang_estimation_1988}.
We begin with the replacement contamination model \citep{denby_robust_1979, bustos_robust_1986, MaYo86}, which takes the form
\begin{equation} \label{eqn:replacement}
y_t = x_t (1 - \delta_t) + \xi_t \delta_t,
\end{equation}
where $\delta_t$ is a binary indicator process taking values in $\{0,1\}$,
and $\xi_t$ denotes an outlying observation whose distribution does not conform to that of $x_t$.
Under this model, the observed series is $y_t$, while $x_t$ remains latent. Given the potentially large magnitude of $\xi_t$, it is convenient to represent contamination in an additive form, leading to
\begin{equation} \label{eqn:additive_outlier}
y_t = x_t + \delta_t \zeta_t,
\end{equation}
where $\zeta_t = \xi_t - x_t$ captures the deviation of the outlier from the latent process.
This specification is referred to as the AO model.

In contrast, under the IO model, the outlier enters through the innovation process and propagates to future observations via the lag structure \citep{fox_outliers_1972, chang_estimation_1988},
\begin{equation} \label{eqn:innovative_outlier}
y_t = x_t + \phi_0^{-1}(L)\,\delta_t \zeta_t.
\end{equation}
In this case, the effect of an outlier is transmitted through the AR dynamics, resulting in persistent deviations over time. Unlike the replacement model \eqref{eqn:replacement}, the IO model therefore allows for time-lagged propagation of contamination.

A key feature of time series models is that contamination at time $t$ affects not only $x_t$, but also the residuals at subsequent times through the dynamic structure of the model.
Let $e_t(\varphi)$ denote the residual computed from the latent uncontaminated process $\{x_t\}$ under a candidate parameter $\varphi \in \mathcal{V}$, so that $e_t(\varphi_0)=\varepsilon_t$.
Let $\tilde e_t(\varphi)$ denote the residual computed from the contaminated observations $\{y_t\}$.
Even when contamination occurs at a single time point, 
the resulting distortion in the residual sequence is generically not local in time.
Because residuals are obtained by filtering the observed series through the model dynamics, contamination entering the data propagates through this filter and affects subsequent residuals.

The following proposition provides an exact characterisation of how contamination propagates through the residual filter.
\begin{proposition}\label{prop:residual}
Under AO contamination, the residual satisfies
\begin{equation}\label{residualao_general}
\tilde e_t(\varphi)
=
e_t(\varphi)
+
\pi(L)\,\delta_t \zeta_t .
\end{equation}
Under IO contamination, the residual satisfies
\begin{equation}\label{residualio_general}
\tilde e_t(\varphi)
=
e_t(\varphi)
+
\pi(L)\phi_0^{-1}(L)\,\delta_t \zeta_t .
\end{equation}
\end{proposition}

The proposition characterises how contamination propagates in the residual process. 
This propagation mechanism is distinct from the more commonly studied propagation in the observed series, which we refer to as \emph{level propagation}.
Level propagation refers to how contamination appears in the observed series itself,
while \emph{residual propagation} refers to how the same contamination appears after the residual filter is applied.
The latter is the directly relevant object for any estimator that can be expressed as a function of the residual sequence. 
Consequently, the observed series alone does not reveal the full effect of contamination on estimation.

\begin{sidewaysfigure}
  \centering
  \caption{Uncontaminated and contaminated series in AR(1), MA(1), and ARMA(1,1) models under AO contamination (left panel) and IO contamination (right panel). Within each panel, rows correspond to the different models, the left column displays the observed series, and the right column displays the corresponding residuals computed using the true model parameters. The shaded band marks the time of the outlier.}
  \includegraphics[width=0.49\textwidth]{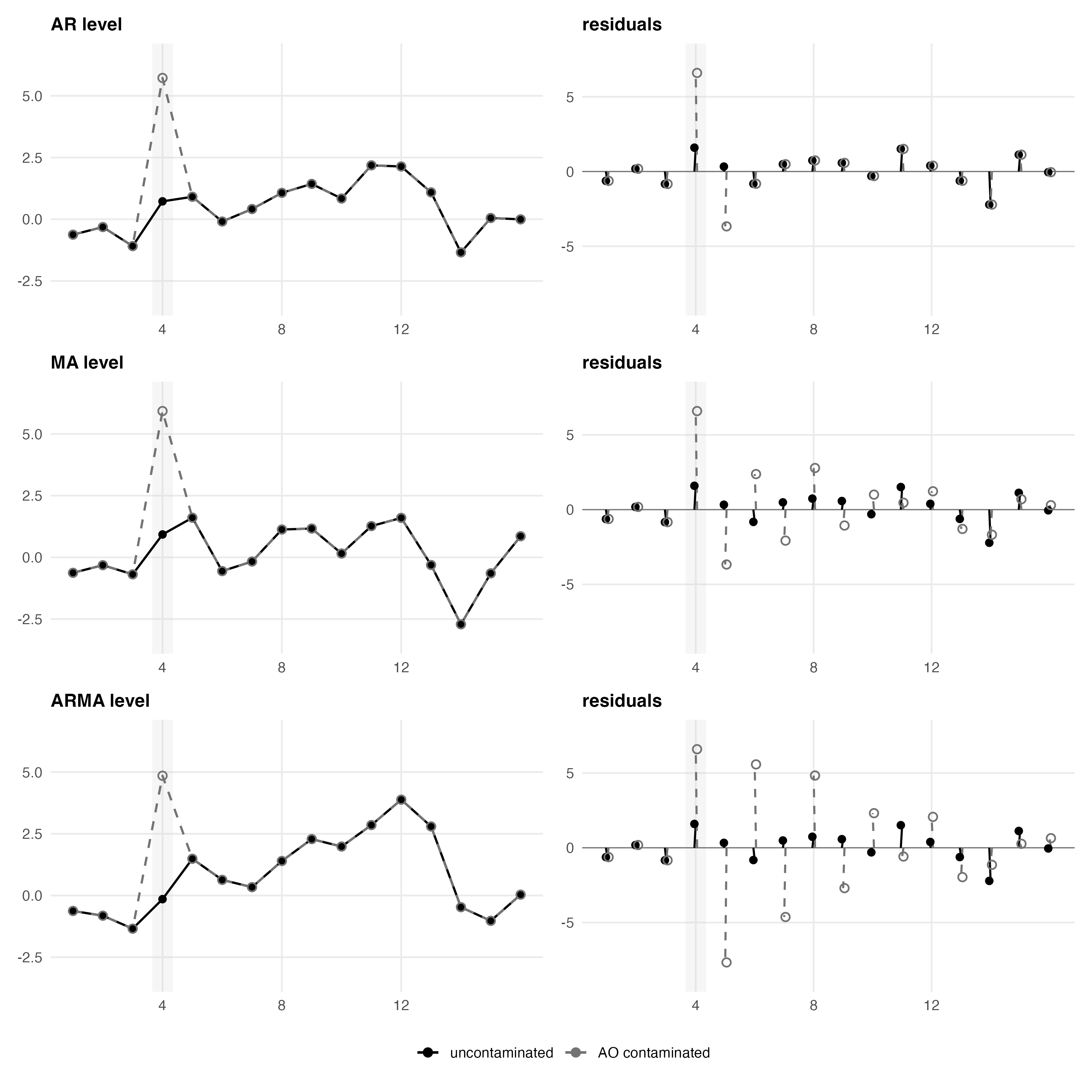}
  \hfill
  \includegraphics[width=0.49\textwidth]{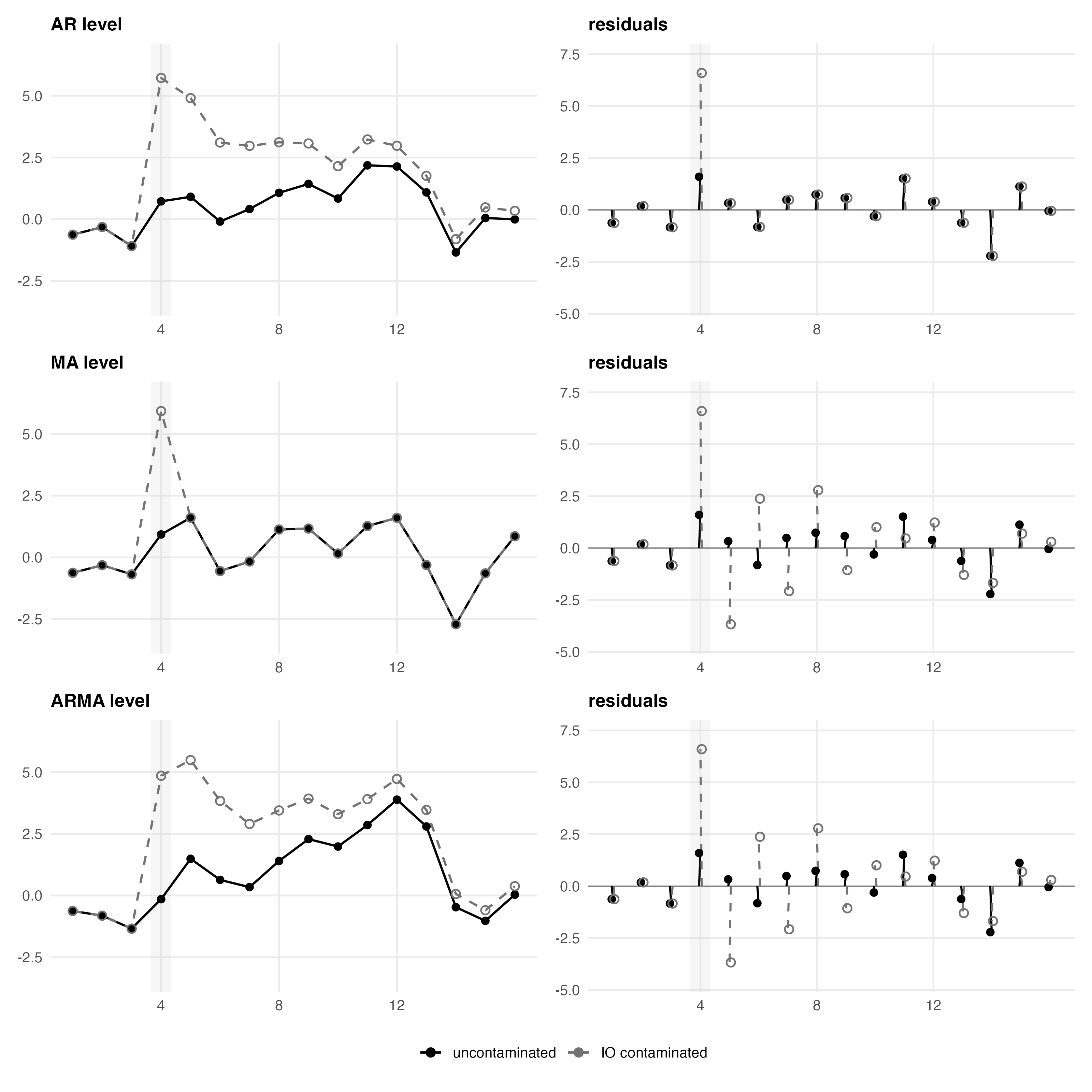}
  \label{fig:AO_IO_combined}
\end{sidewaysfigure}

Figure~\ref{fig:AO_IO_combined} illustrates these propagation mechanisms for AR(1), MA(1), and ARMA(1,1) processes under AO and IO contamination. 
Evaluating the residuals at the true parameter value $\varphi_0$ isolates the contamination mechanism by removing the effect of parameter estimation and provides a benchmark in which parameter uncertainty is absent. 
If residual sequences differ under this benchmark, consistency under contamination cannot be taken for granted.

The propagation patterns described above depend critically on the dynamic structure of the model.
Under the causality and invertibility conditions imposed earlier, the filters $\psi(L)$ and $\pi(L)$ are absolutely summable, so that the effect of a single contamination decays over time in magnitude, even when the associated patch is of infinite length.
This decay property is fundamental for the subsequent analysis, since it implies that residual distortions induced by contamination become asymptotically negligible beyond a sufficiently large horizon.
Without such attenuation, the effect of contamination would persist indefinitely and consistency could not be restored through trimming or patch removal.
Thus, causality and invertibility are not merely technical assumptions, but ensure that contamination propagation through the residual filter remains asymptotically controllable and thereby make consistency restoration possible.

We introduce the notion of patches to formalise the propagation effects described above.
An \emph{outlier patch} corresponds to the consecutive set of time indices affected by the propagation of a single outlier through the time series dynamics.
If the propagation affects the observed series itself, the resulting set of indices is called an \emph{aberrant level patch}, or simply a \emph{level patch}, corresponding to level propagation. 
If the propagation affects the residual sequence, we obtain an \emph{aberrant residual patch}, or \emph{residual patch}, corresponding to residual propagation, defined as a consecutive set of indices for which the residuals do not coincide with the structural errors $e_t(\varphi_0)$ and therefore do not conform to the underlying data-generating mechanism.
In both cases, the patch is generated by dynamic propagation from a single outlier, even though no further outliers are present within the patch itself.

Under AO contamination, the observed series differs from the uncontaminated series only at the outlying observation, so the level effect is local. 
However, the residual filter produces a residual patch in the residual sequence, that is, a set of indices over which the residuals do not coincide with the structural errors $e_t(\varphi_0)$. 
In an AR($p$) model this residual patch has finite length determined by the AR order $p$. 
In contrast, when a moving-average component is present, the inverse moving-average filter generates a residual patch of infinite length.

Under IO contamination the propagation pattern differs. 
In AR models the level effect propagates through the AR dynamics and produces a level patch of infinite length, whereas the residual effect remains local.
In MA models the level effect is local, whereas the residual filter produces a residual patch of infinite length. 
In ARMA models both AR and MA propagation mechanisms operate, so that both level and residual patches are infinite. 
These propagation patterns are summarised in Table~\ref{tab:propagation_patterns}.

\begin{table}[htbp]
\centering
\caption{Propagation patterns and resulting patch lengths}
\label{tab:propagation_patterns}
\begin{tabular}{ccccc}
\hline
Model & AO level & AO residual & IO level & IO residual \\
\hline
AR   & local    & finite    & infinite & local \\
MA   & local    & infinite  & local    & infinite \\
ARMA & local    & infinite  & infinite & infinite \\
\hline
\end{tabular}
\end{table}

Standard robust estimation procedures rely on trimming or down-weighting observations associated with large residuals, see, e.g., \citet{Huber1964, Ro84, Maronna2006}. 
In cross-sectional settings, contamination is local. 
In dynamic time series models, however, Proposition~\ref{prop:residual} shows that contamination propagates through the residual filter and induces non-local distortions in the residual sequence. 
Removing only directly contaminated observations therefore does not eliminate these distortions.

Propagation therefore affects the optimisation problem directly.
Even at the true parameter, contaminated residuals differ from the structural errors underlying the clean criterion, so the usual consistency argument fails.
The figures are evaluated at $\varphi_0$ and therefore represent the most favourable benchmark. 
In practice, $\varphi_0$ is unknown and must be estimated from contaminated data, so the effect of contamination may be more severe.
This motivates the need to control residual propagation in the estimation procedure.

The preceding analysis reveals a structural feature of dynamic contamination.
Residual propagation introduces a non-negligible perturbation into the criterion.
As a consequence, consistency is not a default property and requires additional structure to control the propagated component.
Such control may arise from decay of the underlying filter, sparsity of contamination, or an explicit correction of the retained subset, as developed in the next section.

\section{Subsample-based Estimation under Residual Propagation}\label{sec:estimation}

In this section, we develop a general framework for subsample-based estimation under dynamic contamination. 
The framework applies to a broad class of residual-based subsample estimators and operates at the level of the propagation mechanism. 
The analysis is conducted within a causal and invertible VARMA representation, used as a convenient operator framework in which residual propagation can be characterised explicitly.

As shown in the previous section, contamination propagates through the residual filter and induces non-local distortions in the residual-based criterion. 
This implies a structural requirement for subsample-based estimation in dynamic models.
The retained index set must be compatible with the propagation mechanism, since excluding only directly contaminated observations does not eliminate the resulting residual distortions. 

We formalise this requirement through a criterion-based framework defined on retained subsets of indices. 
Within this framework, we introduce a transformation of index sets that removes the residual footprint of contamination and study the conditions under which it preserves the asymptotic behaviour of the estimator under clean data while restoring consistency under contamination.

\subsection{Residual-based criterion functions}

We consider subsample-based estimators defined through residual-based sample criteria. 
Let $H_T \subset \{1,\dots,T\}$ be a subset of indices with cardinality $h$, and let $\mathcal{H}_T^h$ denote the collection of all such subsets. 
For any $H_T \in \mathcal{H}_T^h$, define
\begin{equation}\label{eq:criterion_subsample}
f_T(\varphi,H_T)
=
g_T\bigl(\{e_t(\varphi):t\in H_T\}\bigr),
\end{equation}
and let $\hat{\varphi}_T(H_T)$ denote a minimiser or maximiser of $f_T(\varphi,H_T)$.
The criterion depends on $\varphi$ only through the residuals $\{e_t(\varphi):t\in H_T\}$.

For simplicity, $\hat{\varphi}_T(H_T)$ is taken to denote an exact optimiser of the sample criterion. 
The same framework can be extended to approximate optimisers that achieve the optimum up to an $o_p(1)$ error, as in the classical theory of extremum estimators; see, for example, \citet{vdV1998}.
The subset $H_T$ may be fixed, data-dependent, or generated by an algorithm. 
The present framework is formulated conditional on a given retained subset $H_T$ and does not depend on how that subset is obtained. 
In particular, the construction of $H_T$ is treated as external to the subsequent estimation problem considered here.
The representation is formulated in terms of the structural residuals $e_t(\varphi)$ associated with the latent process, and thus represents the uncontaminated criterion.

The criterion \eqref{eq:criterion_subsample} encompasses a broad class of estimators used in regression and robust statistics, including least squares \citep{Amemiya1985}, least absolute deviations and quantile regression \citep{KoenkerBassett1978}, maximum likelihood and general M-estimators \citep{Huber1964, NeweyMcFadden1994}, as well as trimming-based procedures such as LTS \citep{Ro84}; see also \citet{Maronna2006}. 
Related residual-based subset selection methods studied by \citet{JohansenNielsen2009,JohansenNielsen2013,JohansenNielsen2016} also fall within this framework.

The analysis is conducted directly at the level of sample criteria and does not rely on the existence of a population objective function or on uniform convergence. 
The key idea is that if the criteria evaluated on two subsets are uniformly close, then, under local stability conditions, the corresponding optimisers are also close. 
This allows us to study subsample transformations without invoking a limiting objective. 
In contrast to standard extremum arguments, which rely on convergence to a deterministic limit with a unique optimiser \citep{Cizek2008}, the present approach operates entirely at the level of the sample criterion.

Under contamination, Proposition~\ref{prop:residual} shows that the residual-based criterion constructed from contaminated observations differs from its uncontaminated counterpart through the propagated residual component.
Unless this propagation effect becomes asymptotically negligible, the contaminated and uncontaminated criteria are not asymptotically equivalent.
Consequently, the corresponding optimisation problems share the same optimiser only under specific conditions for which the propagated residual distortion becomes asymptotically negligible.
Thus, even when the retained subset is asymptotically free of directly contaminated observations, standard subsample-based estimation does not generically yield consistency for the clean-data parameter under dynamic contamination.

\subsection{A propagation-compatible transformation of index sets}

The structural requirement above leads to the following propagation-compatible transformation of retained index sets.
\begin{definition}[Patch removal operator]
For a given integer $\kappa \geq 0$, define the subset transformation
\begin{equation}\label{eqn:patch-removal}
S^\kappa:\, H_T \mapsto 
H_T \setminus
\bigcup_{t \in H_T^c}
\{t+1, \dots, \min(t+\kappa, T)\}.
\end{equation}
This transformation removes, for each excluded index, a forward block of length at most $\kappa$, thereby accounting for the propagation of contamination in the residual sequence.
The transformation $S^\kappa$ induces an operator acting on subsample-based estimators:
\begin{equation}\label{eqn:PROperator1}
K^\kappa:\, \hat{\varphi}_T\, \mapsto\, \hat{\varphi}_T \circ S^\kappa,
\end{equation}
that is,
\begin{equation}\label{eqn:PROperator2}
(K^\kappa \hat{\varphi}_T)(H_T)
=
\hat{\varphi}_T(S^\kappa H_T).
\end{equation}

We refer to $K^\kappa$ as the patch removal operator.
\end{definition}
The transformation $S^\kappa$ is not an additional trimming rule applied for robustness in the usual pointwise sense. 
It is the index-set operation induced by the residual propagation mechanism. 
Once an index is excluded because it may carry contamination, the residual filter implies that later residuals may also be affected. 
A propagation-compatible retained set must therefore exclude the corresponding residual footprint. 
The operator $K^\kappa$ implements this requirement for any subsample-based estimator by composing the original estimator with the transformation $S^\kappa$.

Thus, patch removal operates through the empirical criterion induced by the transformed index set.
In the presence of contamination, it prevents residual patches from entering the criterion. 
Under clean data, it only changes the retained index set and may reduce the effective sample size. 
The invariance result below shows that this cost is asymptotically negligible under suitable stability conditions.

Let $\alpha \in (0,1)$ denote the trimming proportion associated with $H_T$, and let $\kappa$ control the extent of propagation removal. 
Then the transformed subset $S^\kappa H_T$ has cardinality at least $(1 - \alpha(\kappa + 1))T$. 
In the worst case, when residual patches do not overlap, at most $\alpha(\kappa+1)T$ observations are removed. 
In practice, residual patches may overlap, particularly in the presence of consecutive outliers, so that the effective number of removed observations can be smaller.

\subsection{Local stability and asymptotic invariance}

In this subsection, we analyse the optimisation problem defined by the sample criterion $f_T(\varphi, H_T)$ in terms of the structural residuals $e_t(\varphi)$, corresponding to the uncontaminated formulation. 
The analysis is conducted directly at the level of sample criteria and does not rely on convergence to a limiting objective function or differentiability.

The first set of conditions imposes local regularity of the sample criterion around its optimiser and captures properties that arise asymptotically for subsample-based estimators. 
In particular, it ensures separation of the optimiser and local sharpness of the criterion in its neighbourhood.

\begin{assumption}[Local optimisation stability]\label{ass:local_stab}
There exists an open set $\mathcal N \subset \mathcal V$ such that, as $T \to \infty$, with probability tending to one, the following hold.

\noindent
(i) (Local separation)  
There exists $\delta>0$ such that, in the case of minimisation,
\begin{equation}
\inf_{\varphi\notin\mathcal N} f_T(\varphi,H_T)
\ge
\inf_{\varphi\in\mathcal N} f_T(\varphi,H_T) + \delta,
\end{equation}
or, in the case of maximisation,
\begin{equation}
\sup_{\varphi\notin\mathcal N} f_T(\varphi,H_T)
\le
\sup_{\varphi\in\mathcal N} f_T(\varphi,H_T) - \delta.
\end{equation}

\noindent
(ii) (Local sharpness)  
For every $\varepsilon > 0$, there exists $\eta_\varepsilon > 0$ such that,
in the case of minimisation,
\begin{equation}
\inf_{\substack{\varphi \in \mathcal N:\\
\|\varphi-\hat{\varphi}_T(H_T)\|\ge \varepsilon}}
f_T(\varphi,H_T)
\ge
f_T(\hat{\varphi}_T(H_T),H_T) + \eta_\varepsilon,
\end{equation}
or, in the case of maximisation,
\begin{equation}
\sup_{\substack{\varphi \in \mathcal N:\\
\|\varphi-\hat{\varphi}_T(H_T)\|\ge \varepsilon}}
f_T(\varphi,H_T)
\le
f_T(\hat{\varphi}_T(H_T),H_T) - \eta_\varepsilon.
\end{equation}
\end{assumption}

Assumption~\ref{ass:local_stab} formalises local identification and curvature properties of the sample criterion. 
It is analogous to standard conditions used in extremum estimation, and can be verified under conventional sufficient conditions, such as uniform convergence and uniqueness of a limiting objective. 
However, the present framework does not rely on the existence of such a limiting objective and instead operates directly at the level of the sample criterion.

\begin{assumption}[Stability under patch removal]\label{ass:criterion_stability}
The criterion function is uniformly stable under patch removal, that is,
\begin{equation}
\sup_{\varphi\in\mathcal V}
\bigl|
f_T(\varphi,H_T)-f_T(\varphi,S^\kappa H_T)
\bigr|
\overset{p}{\longrightarrow}0.
\end{equation}
\end{assumption}

Assumption~\ref{ass:criterion_stability} depends on how the criterion aggregates residuals. 
Proposition~\ref{prop:patch_stability} provides sufficient conditions illustrating that the high-level condition holds for a wide class of residual-based criteria, in particular when the fraction of removed observations is asymptotically negligible and the criterion contributions are suitably controlled.

Together, Assumptions~\ref{ass:local_stab} and \ref{ass:criterion_stability} ensure that patch removal does not alter the optimisation problem locally. 
Local separation restricts attention to a neighbourhood of the optimiser, stability controls the perturbation induced by the transformation, and sharpness ensures that such perturbations do not shift the optimiser.

The following theorem establishes that, under the stated assumptions, local optimisation stability is preserved under patch removal.
In particular, the criterion based on $S^\kappa H_T$ satisfies the same local separation and sharpness properties as the original criterion.
Furthermore, the patch removal operator is asymptotically invariant, in the sense that the resulting estimator is asymptotically equivalent to the original one.

\begin{theorem}[Stability and invariance under patch removal]
\label{thm:patch_invariance}
Suppose that the data are generated from the uncontaminated model \eqref{eqn:arma}, 
and that Assumptions~\ref{ass:local_stab} and \ref{ass:criterion_stability} hold.
Then, with probability tending to one, the criterion $f_T(\varphi,S^\kappa H_T)$ satisfies the following on the same open set $\mathcal N$.

\noindent
(i) (Transferred local separation)
There exists $\delta'>0$ such that, in the case of minimisation,
\begin{equation}
\inf_{\varphi\notin\mathcal N} f_T(\varphi,S^\kappa H_T)
\ge
\inf_{\varphi\in\mathcal N} f_T(\varphi,S^\kappa H_T)+\delta',
\end{equation}
or, in the case of maximisation,
\begin{equation}
\sup_{\varphi\notin\mathcal N} f_T(\varphi,S^\kappa H_T)
\le
\sup_{\varphi\in\mathcal N} f_T(\varphi,S^\kappa H_T)-\delta'.
\end{equation}

\noindent
(ii) (Transferred local sharpness)
For every $\varepsilon>0$, there exists $\eta'_\varepsilon>0$ such that, in the case of minimisation,
\begin{equation}
\inf_{\substack{\varphi\in\mathcal N:\\
\|\varphi-\hat{\varphi}_T(S^\kappa H_T)\|\ge\varepsilon}}
f_T(\varphi,S^\kappa H_T)
\ge
f_T(\hat{\varphi}_T(S^\kappa H_T),S^\kappa H_T)+\eta'_\varepsilon,
\end{equation}
or, in the case of maximisation,
\begin{equation}
\sup_{\substack{\varphi\in\mathcal N:\\
\|\varphi-\hat{\varphi}_T(S^\kappa H_T)\|\ge\varepsilon}}
f_T(\varphi,S^\kappa H_T)
\le
f_T(\hat{\varphi}_T(S^\kappa H_T),S^\kappa H_T)-\eta'_\varepsilon.
\end{equation}

Furthermore, we have
\begin{equation}
(K^\kappa \hat{\varphi}_T)(H_T) - \hat{\varphi}_T(H_T)
\overset{p}{\longrightarrow}0.
\end{equation}
\end{theorem}

An immediate consequence is that consistency is preserved under patch removal.

\begin{corollary}[Consistency under patch removal]\label{cor:patch_consistency}
Under the conditions of Theorem~\ref{thm:patch_invariance},
if
$\hat{\varphi}_T(H_T)\overset{p}{\longrightarrow}\varphi_0$,
then
\begin{equation}
(K^\kappa \hat{\varphi}_T)(H_T)\overset{p}{\longrightarrow}\varphi_0.
\end{equation}
\end{corollary}

\subsection{Contamination and propagation control}\label{sec:contprop}

We now turn to the contaminated setting and study how contamination affects residual-based criteria. 
We begin with assumptions on the contamination process and then introduce additional regularity conditions on the criterion function.

The purpose of this subsection is not to treat contamination as a structural feature of the data-generating mechanism, but to isolate the conditions under which its effect on the residual-based criterion can be controlled. 
Structural features of the clean stochastic mechanism should not be confused with contamination.
In particular, even when the clean process has no unit root, contamination may generate spurious persistence or unit-root-like behaviour in the observed data. 
Such behaviour does not reflect genuine propagation through the clean model, but rather distortion induced by contamination in the residual sequence.

For this reason, the relevant issue here is not stationarity in a purely descriptive sense, but whether the propagation of contamination through the residual filter can be sufficiently attenuated by patch removal. 
In the present VARMA framework, this is achieved through conditions that yield decay of the propagation effect.
The key requirement is therefore not whiteness of the errors per se, but the existence of a representation under which contamination enters through a propagation mechanism whose effect can be controlled. 
The analysis relies on structural and regularity conditions that ensure that the residual representation is well defined and that this propagation mechanism is suitably attenuated.

\begin{assumption}[Contamination rate]\label{ass:outlier_rate}
There exists $\bar{\alpha} \in [0,1/2)$ such that
\begin{equation}
\limsup_{T\to\infty} T^{-1}\sum_{t=1}^T \delta_t \le \bar{\alpha}.
\end{equation}
\end{assumption}

To ensure that the trimming step is capable of removing contaminated observations asymptotically, the trimming proportion $\alpha$ is assumed to satisfy $\alpha \ge \bar{\alpha}$. 
The restriction $\bar{\alpha}<1/2$ ensures that contaminated observations form a minority of the sample and is standard in robust statistics; see, for example, \citet{Ro84}.

To isolate the effect of residual propagation from the separate problem of outlier detection, we impose the following high-level condition.
In many applications, extreme observations are readily identifiable in the observed series. 
Contamination propagates through the dynamic filter and distorts the residual-based criterion.
\begin{assumption}[Asymptotic absence of contamination]\label{as:ultimate}
\begin{equation}
\Pr\big( \delta_t = 0 \text{ for all } t \in H_T \big) \longrightarrow 1,
\end{equation}
as $T \to \infty$.
\end{assumption}

A sufficient condition, which arises in trimming-based procedures, is that, for sufficiently large $T$,
$\min_{t \in \boldsymbol{\tau}} |\tilde{e}_t(\varphi_0)|
>
\max_{t \notin \boldsymbol{\tau}} |\tilde{e}_t(\varphi_0)|$,
where $\boldsymbol{\tau} = \{t : \delta_t = 1\}$ denotes the set of contaminated indices. 
This enforces strict separation between contaminated and uncontaminated residuals and implies exact classification.

Assumption~\ref{as:ultimate} is substantially weaker. 
It does not require separation of residual magnitudes or exact identification of individual observations. 
Instead, it imposes a one-sided requirement that contaminated observations are asymptotically excluded from the selected subset $H_T$, while uncontaminated observations may be discarded.
Thus false positives are allowed, but false negatives are ruled out with probability tending to one. 
This reflects a conservative selection principle, under which it is preferable to discard additional clean observations in order to eliminate the influence of contamination.

The present paper treats the retained subset $H_T$ as given.
Assumption~\ref{as:ultimate} is imposed as a high-level condition so that the analysis can focus on contamination and residual propagation, rather than on the separate problem of identifying contaminated observations.
Importantly, this condition represents a best-case benchmark. 
Even under this asymptotically uncontaminated subset, standard subsample-based estimation fails in dynamic models, as shown below.

The choice of $\kappa$ is central to the effectiveness of patch removal. 
A common trimming length $\kappa$ provides a uniform upper bound on the propagation of contamination across observations. 
Although the effective propagation length may vary across outliers, the analysis requires only such a bound and does not depend on the exact propagation profile of individual contamination events.

For a causal and invertible VARMA model \eqref{eqn:arma}, the coefficients in the representation $\pi(L)$ decay exponentially fast. 
That is, there exist $M>0$ and $R\in(0,1)$ such that $|\pi_j|\le M R^j$ for all $j$; see \citet[Theorem 11.3.2]{brockwell_time_1991}. 
This implies geometric decay of the contamination effect in the residual sequence. 
Consequently, by choosing $\kappa$ sufficiently large, the residual distortion becomes asymptotically negligible on the retained subset.

If the propagation effect does not decay, as in unit-root-type settings, finite patch removal cannot eliminate the residual footprint of contamination.

\begin{assumption}[Contamination decay condition]\label{ass:contamination_decay}
The contamination magnitude, rate, and trimming level satisfy
\begin{equation}\label{eqn:kappa-condition}
\alpha T \,\|\delta_T \zeta_T\|_\infty \, R^\kappa \longrightarrow 0
\end{equation}
as $T \to \infty$.
\end{assumption}

Assumption~\ref{ass:contamination_decay} can be satisfied by an appropriate choice of $\kappa$. 
Since $R\in(0,1)$, the factor $R^\kappa$ decays exponentially in $\kappa$, whereas $\alpha T\|\delta_T\zeta_T\|_\infty$ may grow with $T$. 
Under standard breakdown-type formulations, $\|\delta_T \zeta_T\|_\infty$ is not assumed to be bounded. 
It therefore suffices for $\kappa$ to increase at a logarithmic rate in $\alpha T\|\delta_T\zeta_T\|_\infty$, so that the exponential decay offsets the combined effect of contamination frequency and magnitude. 
Assumption~\ref{ass:contamination_decay} thus ensures that the residual effect of contamination becomes asymptotically negligible after patch removal.

At the same time, $\kappa$ is constrained by the feasibility condition
\begin{equation}
\alpha(\kappa+1) < 1-c,
\end{equation}
for some constant $c\in(0,1)$, which ensures that the retained subset has cardinality of order $T$. 
If $\kappa$ increases with $T$, the trimming proportion $\alpha$ may need to decrease accordingly, corresponding to settings with a small, possibly vanishing, contamination proportion.

In the VAR($p$) case, propagation has finite memory and is fully eliminated by choosing $\kappa=p$. 
Consequently, once the initial subset is asymptotically free of contamination, no further distortion remains in the residual sequence on the retained subset, and Assumption~\ref{ass:contamination_decay} is no longer required. 
This simplification is formalised in Theorem~\ref{thm:invariance} below.

The assumptions introduced so far concern the contamination mechanism and its interaction with the retained subset. 
We now impose additional regularity conditions on the criterion function itself. 
These conditions are conceptually distinct from the contamination assumptions above. 
They are imposed on the clean residual-based representation and are applied after the effect of contamination has been attenuated on the retained subset. 
Their role is to control the effect of residual perturbations on the criterion.
They are not required for Theorem~\ref{thm:patch_invariance} and Corollary~\ref{cor:patch_consistency}, which are established under the uncontaminated setting.

\begin{assumption}[Criterion regularity]\label{ass:criterion_lipschitz}
For every subset $H_T \in \mathcal H_T^h$, the following conditions hold:

\noindent
(a) The criterion admits the representation
\begin{equation}
f_T(\varphi,H_T)
=
\frac{1}{|H_T|}
\sum_{t\in H_T} m_t\bigl(e_t(\varphi)\bigr).
\end{equation}

\noindent
(b) The functions $m_t:\mathbb R \to \mathbb R$ satisfy
\begin{equation}
|m_t(y)-m_t(x)|
\le
C\bigl(1+|x|+|y-x|\bigr)|y-x|,
\end{equation}
for all $x,y\in\mathbb R$, uniformly in $t$, for some constant $C<\infty$.

\noindent
(c)
\begin{equation}
\sup_{\varphi\in\mathcal V}
\frac{1}{|S^\kappa H_T|}
\sum_{t\in S^\kappa H_T}|e_t(\varphi)|^2
= O_p(1).
\end{equation}
\end{assumption}

Assumption~\ref{ass:criterion_lipschitz} imposes a regularity condition on the criterion function conditional on a given retained subset $H_T$. 
Part (a) restricts attention to criteria that can be written as averages of observation-wise contributions over $H_T$. 
This is standard in subsample-based estimation and is a special case of \eqref{eq:criterion_subsample}, since the criterion depends on $\varphi$ only through the residuals $\{e_t(\varphi): t \in H_T\}$. 
Part (b) imposes a uniform increment condition on the functions $m_t$, requiring that differences in the criterion contributions are controlled linearly in the size of the perturbation, with a coefficient growing at most linearly in the magnitude of the argument. 
This accommodates a broad class of smooth and non-smooth loss functions, including general convex and non-convex M-estimation criteria. 
Part (c) requires that, after patch removal, the average squared residuals over $S^\kappa H_T$ are stochastically bounded uniformly over $\varphi \in \mathcal V$. 
This does not impose bounded residuals, nor does it require the same property to hold over the full sample or the original subset $H_T$. 
It allows for occasional large values and is compatible with heteroskedasticity.

Taken together, these conditions are weak enough to cover most commonly used criterion functions in time series and econometrics, while ensuring that perturbations in the residuals induce controlled perturbations of the criterion. 
The assumption operates directly at the level of criterion increments, rather than imposing unnecessary smoothness conditions.

Let $\tilde{\varphi}_T(H_T)$ denote an estimator obtained by optimising a sample criterion function of the form
\begin{equation}\label{eq:criterion_consubsample}
\tilde f_T(\varphi,H_T)
=
g_T\bigl(\{\tilde e_t(\varphi):t\in H_T\}\bigr),
\end{equation}
where, as in \eqref{eq:criterion_subsample}, the criterion depends on $\varphi$ only through the contaminated residuals $\{\tilde e_t(\varphi):t\in H_T\}$.

\begin{theorem}[Invariance under contamination]
\label{thm:invariance}
Suppose that Assumptions~\ref{ass:local_stab}--\ref{ass:criterion_lipschitz} hold.
Then, under AO or IO contamination in model \eqref{eqn:arma},
\begin{equation}
(K^\kappa \tilde{\varphi}_T)(H_T) - (K^\kappa \hat{\varphi}_T)(H_T)\overset{p}{\longrightarrow}0.
\end{equation}
Moreover, in the VAR($p$) case with $\kappa = p < \infty$, the result holds without Assumption~\ref{ass:contamination_decay}.
\end{theorem}

An immediate consequence is that consistency is preserved under contamination.

\begin{corollary}[Consistency under contamination]\label{cor:consistency}
Under the conditions of Theorem~\ref{thm:invariance}, if
$\hat{\varphi}_T(H_T)\overset{p}{\longrightarrow}\varphi_0$,
then
\begin{equation}
(K^\kappa \tilde{\varphi}_T)(H_T)\overset{p}{\longrightarrow}\varphi_0.
\end{equation}
\end{corollary}

The preceding results do not require $\kappa$ to be fixed and remain valid for sequences $\kappa=\kappa_T$, provided that the assumptions continue to hold. 
This yields a useful separation in the analysis. 
The optimisation results are independent of the contamination mechanism, while the roles of $\alpha$ and $\kappa$ arise only through model-specific conditions, such as Assumption~\ref{ass:contamination_decay}, that ensure stability.

The analysis is conducted within a weakly stationary VARMA framework, which provides a convenient operator representation for studying residual propagation. 
Under weak stationarity, Wold's decomposition theorem \citep{wold1938study} ensures that the process admits an infinite-order VMA representation. 
Thus, the propagation mechanism is not specific to finite-order VARMA models, but arises more generally in dynamic systems admitting linear or approximately linear filtering representations. 
Under standard conditions, such representations can be approximated by causal and invertible VARMA models, making the propagation structure explicit through the associated filters. 
For this reason, the VARMA framework is sufficiently general to capture a broad class of weakly stationary linear dynamic systems relevant for the present analysis.

Stationarity itself, however, is not essential. 
What matters is whether the propagation effect induced by contamination decays sufficiently fast on the retained subset. 
In the present framework, this is ensured through Assumption~\ref{ass:contamination_decay}. 
When propagation does not decay, finite patch removal is insufficient.

Similar propagation effects also arise in many weakly stationary nonlinear models, including heteroskedastic and regime-switching processes. 
In such settings, the clean-data dynamics themselves may be nonlinear, while contamination remains an external perturbation acting through the residual filter.

The framework admits the boundary case $\kappa=0$, in which
$S^\kappa H_T = H_T$,
so that patch removal reduces to the identity operator.
In this case Assumption~\ref{ass:contamination_decay} becomes
\begin{equation}\label{eq:strongcase}
\alpha T\|\delta_T\zeta_T\|_\infty \longrightarrow 0.
\end{equation}
This condition is highly restrictive and excludes settings in which residual propagation induces a non-negligible distortion of the criterion.
Thus the case $\kappa=0$ does not contradict the generic inconsistency of subsample-based estimation under dynamic contamination.
Rather, it corresponds to an exceptional boundary case in which propagation-compatible patch removal is asymptotically unnecessary.

\section{Verification of the assumptions}\label{sec:verification}

In this section, we clarify the roles of the assumptions and indicate which of them require verification in applications.
Notably, the analysis does not require structural assumptions on the contamination process, such as independence, distributional forms, or temporal dependence.
Instead, contamination enters only through its magnitude, its rate, and its propagation through the residual filter.

The propositions below focus primarily on Assumptions~\ref{ass:local_stab} and \ref{ass:criterion_stability}, which are high-level conditions imposed directly on the optimisation problem and whose generality may not be immediate from their formulation.
By contrast, Assumptions~\ref{ass:outlier_rate}--\ref{ass:contamination_decay} depend on the contamination and screening mechanism, while Assumption~\ref{ass:criterion_lipschitz} is a standard regularity condition satisfied by a broad class of residual-based criteria.
The propositions provide sufficient conditions under which Assumptions~\ref{ass:local_stab} and \ref{ass:criterion_stability} hold, thereby connecting the present framework to familiar extremum estimation arguments and existing subsample-based robust procedures.
The conditions imposed throughout are sufficient rather than necessary and are formulated at a high level to capture the structural features of the problem across a broad class of settings, rather than to characterise minimal assumptions in each component.

We begin by connecting the local optimisation stability condition to the classical extremum estimator framework.
The result is stated for minimisation, and the maximisation case is analogous.

\begin{proposition}
\label{prop:extremum_to_localstab}
Consider the uncontaminated VARMA model \eqref{eqn:arma}. Fix a subset
$H_T \in \mathcal H_T^h$, and let
\begin{equation}
\hat{\varphi}_T(H_T)
=
\arg\min_{\varphi\in\mathcal V} f_T(\varphi,H_T),
\end{equation}
where $\mathcal V \subset \mathbb R^m$ is compact and $\varphi_0 \in \mathcal V$.
Assume that

\noindent
(i) the mapping
$\varphi \mapsto f_T(\varphi,H_T)$
is continuous;

\noindent
(ii) there exists a deterministic function $f:\mathcal V\to\mathbb R$ such that
\begin{equation}
\sup_{\varphi\in\mathcal V}
\bigl|f_T(\varphi,H_T)-f(\varphi)\bigr|
\overset{p}{\longrightarrow}0;
\end{equation}

\noindent
(iii) the function $f$ is continuous and has a unique minimiser at $\varphi_0$.

Then
(a) $\hat{\varphi}_T(H_T)\overset{p}{\longrightarrow}\varphi_0$;
(b) Assumption~\ref{ass:local_stab} holds.
\end{proposition}

Proposition~\ref{prop:extremum_to_localstab} shows that Assumption~\ref{ass:local_stab} is compatible with the standard extremum-estimator framework. 
In particular, whenever the sample criterion converges uniformly to a deterministic limit with a uniquely identified minimiser, local separation and local sharpness follow asymptotically. 
Thus, the present framework includes the classical consistency setting as a special case, even though the main results of this paper are formulated directly at the level of sample criteria and do not rely on the existence of a limiting objective function.

\begin{proposition}
\label{prop:patch_stability}
Suppose that Assumption~\ref{ass:criterion_lipschitz}(a) holds for $H_T$ and the corresponding transformed subsets $S^\kappa H_T$ for all sufficiently large $T$.
Assume that

\noindent
(i) the number of removed indices is asymptotically negligible, that is,
\begin{equation}
\frac{|H_T\setminus S^\kappa H_T|}{|H_T|}
\overset{p}{\longrightarrow}0;
\end{equation}

\noindent
(ii) the criterion contributions are uniformly bounded in probability,
\begin{equation}
\sup_{\varphi\in\mathcal V}\sup_{t\in H_T}
\bigl|m_t(e_t(\varphi))\bigr|
=
O_p(1).
\end{equation}

Then Assumption~\ref{ass:criterion_stability} holds.
\end{proposition}

Proposition~\ref{prop:patch_stability} provides a transparent route to verify Assumption~\ref{ass:criterion_stability} for criteria constructed as averages of residual contributions. 
The key requirement is that patch removal affects only an asymptotically negligible fraction of the retained indices, so that the induced perturbation of the empirical criterion vanishes uniformly in the parameter. 
Formally, this corresponds to condition (i).

To interpret this condition, note that $S^\kappa$ removes, for each initially excluded observation, a patch of length at most $\kappa$. 
Hence, the total number of additionally removed observations is at most of order $\alpha \kappa T$, while $|H_T| \approx (1-\alpha)T$, yielding the bound
\begin{equation}
\frac{|H_T\setminus S^\kappa H_T|}{|H_T|}
\lesssim
\frac{\alpha \kappa}{1-\alpha}.
\end{equation}
A sufficient condition is therefore that $\alpha \kappa \to 0$. 
Combined with the contamination decay requirement in Assumption~\ref{ass:contamination_decay}, which typically requires $\kappa$ to grow at a logarithmic rate, this leads to conditions of the form
\begin{equation}
\alpha \log(\alpha T \|\delta_T \zeta_T\|_\infty) \to 0,
\end{equation}
and, in typical settings, to $\alpha \log T \to 0$. 
This corresponds to a setting in which contamination is asymptotically sparse, even though the total number of contaminated observations may diverge.

Condition (i) is not the only possible route to stability. 
An alternative is to impose conditions under which the criteria based on $H_T$ and $S^\kappa H_T$ have asymptotically equivalent behaviour under clean data. 
However, such an approach typically requires stronger assumptions on the data-generating process and dependence structure, since the thinning induced by $S^\kappa$ must preserve the limiting behaviour of the criterion. 
The advantage of condition (i) is that it achieves stability by directly controlling the fraction of affected observations, making it particularly suitable in settings with general dependence structures.

For pure VAR$(p)$ models, the effective patch length is finite, so that $\kappa$ can be taken as a fixed constant (e.g.\ $\kappa=p$), rather than growing with $T$. 
In this case, the condition $\alpha \kappa \to 0$ reduces to $\alpha \to 0$, and the compatibility with Assumption~\ref{ass:contamination_decay} becomes immediate.

Condition (ii) is stated as a convenient sufficient condition. 
For unbounded criteria, such as quadratic or Gaussian likelihood contributions, the same conclusion can be obtained under suitable moment and maximal control conditions on the residual sequence.

The preceding propositions show that Assumptions~\ref{ass:local_stab} and \ref{ass:criterion_stability} can be verified under familiar conditions. 
This clarifies the scope of the present framework and its connection to existing subsample-based robust estimators. 
In particular, the framework is compatible with the robust procedures studied by \citet{JohansenNielsen2009,JohansenNielsen2013,JohansenNielsen2016}, insofar as these procedures produce a retained subset $H_T$ on which subsequent estimation is based. 
Procedures such as Huber-skip and LTS differ in how the subset $H_T$ is constructed, but this distinction is immaterial here, as the framework is formulated conditional on $H_T$ and does not depend on how it is obtained. 
Different procedures may employ one-step rules, such as Huber-skip and LTS, or iterative and sequential methods such as the forward search \citep{HadiSimonoff1993, AtkinsonRiani2000}. 
Once a subset $H_T$ is obtained, the patch removal operator can be applied to construct $S^\kappa H_T$, after which estimation proceeds on the transformed subset.

Under the regularity conditions imposed in that literature, including tightness of initial estimators and conditions on the innovation density and regressors, the sample criteria underlying these estimators satisfy the sufficient conditions in Propositions~\ref{prop:extremum_to_localstab} and \ref{prop:patch_stability}. 
These results therefore provide concrete examples under which Assumptions~\ref{ass:local_stab} and \ref{ass:criterion_stability} hold jointly.

The high-level assumptions are thus weaker than the sufficient conditions in Propositions~\ref{prop:extremum_to_localstab} and \ref{prop:patch_stability}, which are themselves mild and hold under familiar regularity conditions. 
Together with Theorem~\ref{thm:invariance}, this implies that, once residual propagation is properly controlled, the proposed transformation ensures consistent estimation for a broad class of procedures based on subset selection in dynamic models under sufficiently sparse contamination.

\section{Simulation}\label{sec:simulation}

This section studies the finite-sample implications of the propagation-compatible correction under controlled contamination. 
The design isolates the estimation problem from the subset selection problem. 
In standard robust procedures, methods such as Huber-type estimators, LTS, or related approaches produce a retained subset of observations to be used for subsequent estimation. 
While these procedures differ in how the subset $H_T$ is constructed, the present analysis treats that subset as given and studies the effect of residual propagation and patch removal conditional on it.

In the present design, we therefore work directly with a retained subset $H_T$ containing only uncontaminated observations. 
This isolates residual propagation and patch removal from the separate problem of outlier detection and represents a best-case benchmark for subsample-based procedures.

In each Monte Carlo replication, the model parameters are re-estimated from the simulated sample using a residual-based subsample estimator defined on the retained set. 
Since all such procedures share the same second-stage estimation on $H_T$, the specific choice of estimator is not central to the purpose of the experiment.
Estimation is carried out by numerically minimising the determinant of the residual covariance matrix over the selected observations after patch removal. 
The optimisation is initialised at the true parameter value to isolate the effect of contamination and propagation correction from numerical issues. 
The model order is assumed to be known and correctly specified.

We consider three standard models VAR(1), VMA(1), and VARMA(1,1). 
Let $\varepsilon_t \sim N_2(0, \Sigma)$ and
\begin{equation}
\Sigma =
\begin{pmatrix}
1 & 0.2 \\
0.2 & 1
\end{pmatrix},
\quad
A =
\begin{pmatrix}
0.7 & 0 \\
0.3 & 0.7
\end{pmatrix}.
\end{equation}
The data-generating processes are
\begin{equation}
\begin{aligned}
\text{VAR(1):} \quad & x_t = A x_{t-1} + \varepsilon_t, \\
\text{VMA(1):} \quad & x_t = \varepsilon_t + A \varepsilon_{t-1}, \\
\text{VARMA(1,1):} \quad & x_t = A x_{t-1} + \varepsilon_t + A \varepsilon_{t-1}.
\end{aligned}
\end{equation}
Contamination is introduced as AO or IO. 
A proportion $\alpha \in \{1\%, 5\%, 10\%\}$ of indices is contaminated with magnitude $\zeta \in \{5, 10, 50, 100\}$. 
Sample sizes are $T \in \{500, 1000\}$.

The trimming parameter $\kappa$ controls the length of residual patches removed after each contaminated observation. 
The case $\kappa=0$ corresponds to oracle subsampling without accounting for residual propagation, and serves as a benchmark for the failure of standard subsample-based procedures. 

Estimator performance is evaluated using total bias and root mean squared error (RMSE)
\begin{equation}
\text{Bias} =
\sqrt{
\sum_{j=1}^n
\left(
\frac{1}{m} \sum_{i=1}^m \tilde{\varphi}_j^{(i)} - \varphi_{0,j}
\right)^2
}, \quad
\text{RMSE} =
\sqrt{
\sum_{j=1}^n
\left(
\frac{1}{m} \sum_{i=1}^m
\left(
\tilde{\varphi}_j^{(i)} - \varphi_{0,j}
\right)^2
\right)
}.
\end{equation}
The simulation results are reported in Tables~\ref{tab:VAR_bias_RMSE}--\ref{tab:VARMA_bias_RMSE_1000}. 
Across all experiments, a clear failure-correction contrast emerges. 
When $\kappa=0$, substantial bias arises under contamination despite oracle subsampling, confirming that subsample-based estimation fails even when the retained subset contains only uncontaminated observations. 
This shows that the failure is not due to imperfect identification of contaminated observations, but arises from residual propagation in the estimation stage itself. 
The distortion increases sharply with the outlier magnitude $\zeta$, reflecting amplification through the residual-based criterion. 
In contrast, increasing $\kappa$ restores estimation accuracy, while inducing only a mild efficiency loss under clean data.

Table~\ref{tab:VAR_bias_RMSE} reports the VAR(1) results. 
Under AO contamination, estimation deteriorates sharply at $\kappa=0$, reflecting finite residual propagation that is not removed by standard trimming. 
Setting $\kappa=1$ restores performance to the clean-data benchmark, consistent with the finite propagation length. 
IO contamination has no effect in this setting and is omitted, in line with Theorem~2.

Table~\ref{tab:VMA_bias_RMSE} reports the VMA(1) results. 
Under contamination, severe distortions arise at $\kappa=0$, reflecting the accumulation of propagated residual effects. 
As $\kappa$ increases, performance improves markedly, with particularly strong gains for large $\zeta$. 
AO and IO yield identical results, as both generate non-local residual effects.

Tables~\ref{tab:VARMA_bias_RMSE_500} and \ref{tab:VARMA_bias_RMSE_1000} report the VARMA(1,1) results. 
Under contamination, both AO and IO induce substantial distortions at $\kappa=0$, with AO having a larger initial effect. 
Increasing $\kappa$ reduces the distortions markedly and brings both cases close to the clean benchmark. 
The difference between AO and IO diminishes for sufficiently large $\kappa$, reflecting the joint propagation structure.

Taken together, the simulation results confirm the theoretical findings and clarify their implications. 
First, even under oracle subsampling, substantial bias arises when residual propagation is not controlled, demonstrating that subsample-based estimation fails in dynamic models. 
Second, the distortion increases with the magnitude of contamination, reflecting its amplification through the residual-based criterion. 
Third, patch removal restores estimation accuracy across all settings, while incurring only a modest efficiency loss under clean data. 
These findings highlight that reliable estimation requires controlling residual propagation, rather than relying solely on pointwise trimming.

\section{A Feasible Implementation of the Patch Removal Operator} \label{sec:algorithm}

This section describes a feasible implementation of the patch removal operator for empirical use and illustrates how it can be incorporated into existing robust procedures.
Subsample-based robust procedures for contaminated time series produce a retained subset for estimation; see, for example, \citet{JohansenNielsen2016}.
Suppose that such a procedure produces an estimate and a retained subset $H_T$. 
Patch removal then refines $H_T$ to account for residual propagation.

Given an initial subset $H_T$, define the patch-removed subset $S^\kappa H_T$ according to \eqref{eqn:patch-removal}, where $\kappa \ge 0$ is a user-specified patch length. 
The estimator is then re-applied to the reduced subset, yielding
\begin{equation}
\tilde{\varphi}_T^{\kappa}
=
(K^\kappa \tilde{\varphi}_T)(H_T)
=
\tilde{\varphi}_T(S^\kappa H_T).
\end{equation}
This step implements the operator $K^\kappa$ and removes not only directly contaminated observations but also the subsequent residual patches induced by dynamic propagation. 
The screening step serves only to provide an initial conservative subset. 

In finite samples, the initial screening step may misclassify observations, particularly when it relies on level-based diagnostics. 
For this reason, the procedure may be iterated. 
Let $\tilde{\varphi}_T^{(\kappa,0)}=\tilde{\varphi}_T(H_T)$ and $H_T^{(0)}=H_T$. 
The recursive scheme is given in Algorithm~\ref{alg:iteration}. 
Under the conditions of Theorem~\ref{thm:invariance}, a single application of $K^\kappa$ is sufficient for asymptotic validity, so iteration serves only as a finite-sample refinement.
When the initial screening step is accurate, the algorithm typically stabilises quickly. 
When classification is imperfect, the refinement step may improve the approximation to the uncontaminated subset.

\begin{algorithm}
\caption{Iterative feasible patch removal procedure}
\label{alg:iteration}
\begin{algorithmic}[1]
\Require Initial retained subset $H_T$, robust estimator $\tilde{\varphi}_T(\cdot)$, patch length $\kappa \ge 0$
\Ensure Final estimate $\tilde{\varphi}_T^{(\kappa,m)}$ and retained subset $H_T^{(m)}$

\State Initialise $\tilde{\varphi}_T^{(\kappa,0)}=\tilde{\varphi}_T(H_T)$ and $H_T^{(0)}=H_T$.

\For{$m=0,1,2,\dots$}
    \State Compute residuals $\tilde e_t^{(m)}$ based on $\tilde{\varphi}_T^{(\kappa,m)}$, for $t=1,\dots,T$.
    \State Apply the estimator-specific screening rule based on $\{\tilde e_t^{(m)}\}_{t=1}^T$ to obtain $H_T^{(m+1)}$.
    \State Compute $\tilde{\varphi}_T^{(\kappa,m+1)} = (K^\kappa \tilde{\varphi}_T)(H_T^{(m+1)}) = \tilde{\varphi}_T(S^\kappa H_T^{(m+1)})$.
    \If{$H_T^{(m+1)} = H_T^{(m)}$}
        \State \textbf{return} $\bigl(\tilde{\varphi}_T^{(\kappa,m+1)}, H_T^{(m+1)}\bigr)$
    \EndIf
\EndFor
\end{algorithmic}
\end{algorithm}

The choice of $\kappa$ is governed by the propagation structure of the underlying dynamic model. 
In particular, the effective propagation length depends on the AR and moving-average orders $(p,q)$ and on the contamination mechanism. 
This connection is formalised in Assumption~\ref{ass:contamination_decay}. 
For pure VAR($p$) models, residual patches have finite length, so that $\kappa$ can be taken equal to $p$. 
For models with a moving-average component, residual patches are of infinite length but decay geometrically, so that $\kappa$ must typically increase with the sample size at a logarithmic rate. 
Thus, the choice of $\kappa$ is linked to the underlying VARMA specification.

In practice, however, neither the model orders $(p,q)$ nor the contamination mechanism are known in advance, and both may be difficult to infer reliably in contaminated data. 
This creates a circular problem in that the choice of $\kappa$ depends on the model specification, while the model specification itself is obscured by contamination.
We therefore adopt a pragmatic joint-selection approach. 
A collection of candidate specifications indexed by $(\kappa,p,q)$ is considered. 
For each candidate, $\kappa$ is chosen large enough to remove the dominant part of the residual patch, but not so large that the retained subset becomes too small. 
The estimation procedure in Algorithm~\ref{alg:iteration} is then applied to each candidate.

Candidate specifications may be screened using the observed propagation patterns in the contaminated series. 
Although uncontaminated residuals are not observable, the behaviour of observed aberrant level and residual patches still provides useful diagnostic information. 
In particular, Table~\ref{tab:propagation_patterns} shows that AO contamination does not generate persistent level patches, but may induce persistent residual patches, whereas IO contamination generates persistent level patches together with model-dependent residual patch behaviour. 
These qualitative features can be used to rule out implausible combinations of contamination type, patch length, and dynamic specification.

For estimation itself, explicit classification of contamination as AO or IO is not essential. 
The role of $\kappa$ is to remove the dominant part of residual propagation, and a sufficiently conservative choice typically yields robustness regardless of the underlying mechanism. 
The AO--IO distinction is more relevant for interpretation and prediction, since different contamination mechanisms imply different propagation patterns and hence potentially different out-of-sample implications. 
For this reason, the distinction is used here mainly as a diagnostic device and is revisited in the empirical illustration.

After this preliminary reduction of the candidate set, we require a simple rule to compare the remaining models. 
The robust model selection literature contains substantially more elaborate proposals. 
For example, \citet{Müller01122005} combine a robust penalised in-sample criterion with a robust estimate of conditional prediction loss obtained via a stratified bootstrap, and note that the associated computational burden can be substantial. 
Model selection, however, is not the primary objective of the present paper. 
Our purpose is only to adopt a simple and operational comparison rule for illustration.

We therefore employ a normalised information criterion based on the average log-likelihood, that is, AIC on a per-observation scale:
\begin{equation}\label{eqn:aic_avg}
\mathrm{AIC}_{\mathrm{avg}}
=
-\frac{2}{|S^\kappa H_T|}\log \hat L
+
\frac{2k}{|S^\kappa H_T|},
\end{equation}
where $|S^\kappa H_T|$ is the effective sample size after subsampling and patch removal, and $k$ is the number of estimated parameters. 
This normalisation reduces the scale effect induced by differing effective sample sizes and facilitates heuristic comparison across candidates. 
The criterion is used here solely as a practical device rather than as a fully developed contribution to robust model selection.

\section{Empirical Illustration}\label{sec:illustration}

We consider two daily Icelandic river flow series covering the period from 1 January 1972 to 31 December 1974, yielding a total of $T=1096$ observations, in which the flow is measured in cubic metres per second. The data originate from the Hydrological Survey of the National Energy Authority of Iceland and were first analysed by \citet{ToThGu85}. The dataset has subsequently been used in the time series literature, including \citet{Tsa98} and \citet{TeYa14}, where bivariate river flow systems are analysed using threshold and smooth transition models, respectively.
The two series correspond to the rivers Jökulsá eystra and Vatndalsá. The former is the larger river, with a drainage basin that includes a glacier, whereas the latter has a smaller basin with a substantial contribution from groundwater. This leads to smoother dynamics in the former and more pronounced variation in the latter.

In the present analysis, we focus on the observed bivariate river flow series and denote them by $y_{1t}$ and $y_{2t}$. The series exhibit strong serial correlation and a nontrivial temporal structure, making them suitable for illustrating the framework developed in this paper. Unlike earlier studies, we do not incorporate exogenous variables such as precipitation or temperature, as our aim is not to construct a structural model of river flow dynamics.

We do not impose any assumption on the presence of outliers in the data. The empirical illustration is designed to examine the behaviour of the proposed procedure under potential contamination, irrespective of whether such contamination is present in the observed series. The analysis is therefore conducted in a robustness framework rather than as a diagnostic exercise.
While the dataset has often been used to study nonlinear and regime-dependent dynamics, our objective here is different in that it is used to provide a realistic temporal structure within which the impact of contamination on residual-based estimation can be examined.

Accordingly, the dataset serves as a representative example of a multivariate time series with nontrivial autocorrelation. This setting allows us to study how contamination, if present, propagates through temporal dependence and affects residual-based estimation, and to assess how the proposed procedure mitigates such effects in practice. For reference, Figure~\ref{fig:river_flags} displays the series together with observations identified by the iterative procedure in Algorithm~\ref{alg:iteration}.

Model comparison is carried out over a range of VARMA$(p,q)$ specifications. For each candidate model, estimation is conducted using the Huber-skip procedure with trimming proportion $\alpha=0.1$, followed by patch removal with $\kappa=5$. Model comparison is based on a Gaussian log-likelihood computed from the residuals, and we report the corresponding $\mathrm{AIC}_{\mathrm{avg}}$ \eqref{eqn:aic_avg} in Table~\ref{tab:aic_avg}.

\begin{table}[ht]
\centering
\begin{tabular}{ccccc}
\toprule
$p$ & $q$ & $|S^\kappa H_T|$ & $\mathrm{AIC}_{\mathrm{avg}}^{\text{patch}}$ & $\mathrm{AIC}_{\mathrm{avg}}^{\text{full}}$ \\
\midrule
1 & 0 & 730 & -5.28034 & -2.99327 \\
2 & 0 & 715 & \textbf{-5.38266} & -3.08275 \\
3 & 0 & 714 & -5.16179 & -3.00439 \\
1 & 1 & 725 & -5.34597 & -3.08562 \\
2 & 1 & 716 & -5.34479 & \textbf{-3.15690} \\
3 & 1 & 709 & -5.11340 & -2.87621 \\
1 & 2 & 716 & -5.37751 & 269.99579 \\
2 & 2 & 723 & -5.12414 & -3.05836 \\
3 & 2 & 724 & -4.68647 & -2.87460 \\
1 & 3 & 719 & 174.12730 & -2.86583 \\
2 & 3 & 987 & -2.86049 & 94.37901 \\
3 & 3 & 709 & -4.82302 & -2.92548 \\
\bottomrule
\end{tabular}
\caption{Model comparison based on $\mathrm{AIC}_{\mathrm{avg}}$ for different $(p,q)$ with $\kappa=5$.}
\label{tab:aic_avg}
\end{table}

We restrict attention to a single subsample-based procedure for illustration. While alternative approaches, such as LTS, could also be employed, the focus is not on comparing robust estimation algorithms, but on isolating the effect of residual propagation and patch removal on the estimation criterion. The Huber-skip procedure is therefore used as a representative method that produces a data-dependent subset and permits repeated estimation across candidate models.
For each specification, we report the full sample estimator, the Huber-skip estimator without patch removal, and the corresponding estimator with patch removal.
These estimators differ in how contamination is handled.
The full sample estimator uses the observed series directly, Huber-skip applies pointwise trimming, and patch removal additionally accounts for residual propagation.
The comparison between Huber-skip and patch removal therefore isolates the effect of residual propagation, as any remaining discrepancy after trimming reflects contamination effects that are not removed by standard robust procedures.

The trimming proportion $\alpha=0.1$ is chosen to allow for a moderate level of contamination, so that the retained subset is expected to be dominated by observations consistent with the underlying data-generating mechanism. For each $(p,q)$, we report two versions of the information criterion, one based on the retained subset after patch removal and one based on the full sample. The use of $\mathrm{AIC}_{\mathrm{avg}}$ is heuristic, serving as a per-observation measure for comparing candidate models rather than as a formal robust model selection procedure.

Several features emerge from the results. First, the patch removal procedure generally yields lower values of $\mathrm{AIC}_{\mathrm{avg}}$, reflecting an improved Gaussian fit once extreme observations and their propagation effects are removed. Second, for certain specifications, such as $(p,q)=(1,2)$ and $(2,3)$, the information criterion becomes extremely large under both approaches, indicating numerical instability in the estimated residual covariance matrix. This behaviour is not driven by subsampling or patch removal, but reflects the well-known sensitivity of VARMA models with MA components, which require nonlinear optimisation and may yield ill-conditioned estimates in finite samples. 

A related feature arises for $(p,q)=(2,3)$, where $|S^\kappa H_T| = |H_T|$, indicating that identified extreme observations are concentrated at the end of the sample, so that patch removal becomes inactive. Taken together, these patterns suggest that irregular values for higher-order specifications are driven primarily by numerical difficulties associated with the MA component, rather than by the subsampling or patch removal procedure. This distinction highlights that the proposed method primarily affects the statistical properties of the estimation criterion, whereas numerical stability is governed by the underlying VARMA structure.

Finally, the model minimising $\mathrm{AIC}_{\mathrm{avg}}$ differs across the two approaches. The VARMA$(2,0)$ specification minimises $\mathrm{AIC}_{\mathrm{avg}}$ under patch removal, whereas the VARMA$(2,1)$ model achieves the lowest value under the full sample criterion. This discrepancy illustrates that departures from Gaussianity, including the presence of extreme observations, can materially affect likelihood-based model comparison.

We report all model specifications without selection or filtering to provide a transparent assessment across configurations, including cases exhibiting numerical instability. The grid is restricted to $p \geq 1$, and no improvement is observed from higher-order specifications. Overall, the results should be interpreted as illustrating how the proposed procedure modifies likelihood-based model comparison, rather than as definitive evidence for a particular model choice.

\begin{figure}
  \centering
  \caption{Daily river flow series for Jökulsá (top) and Vatndalsá (bottom).
  Vertical lines indicate observations identified as outliers by Algorithm~\ref{alg:iteration} with $(\kappa = 2)$.}
  \includegraphics[width=0.9\textwidth]{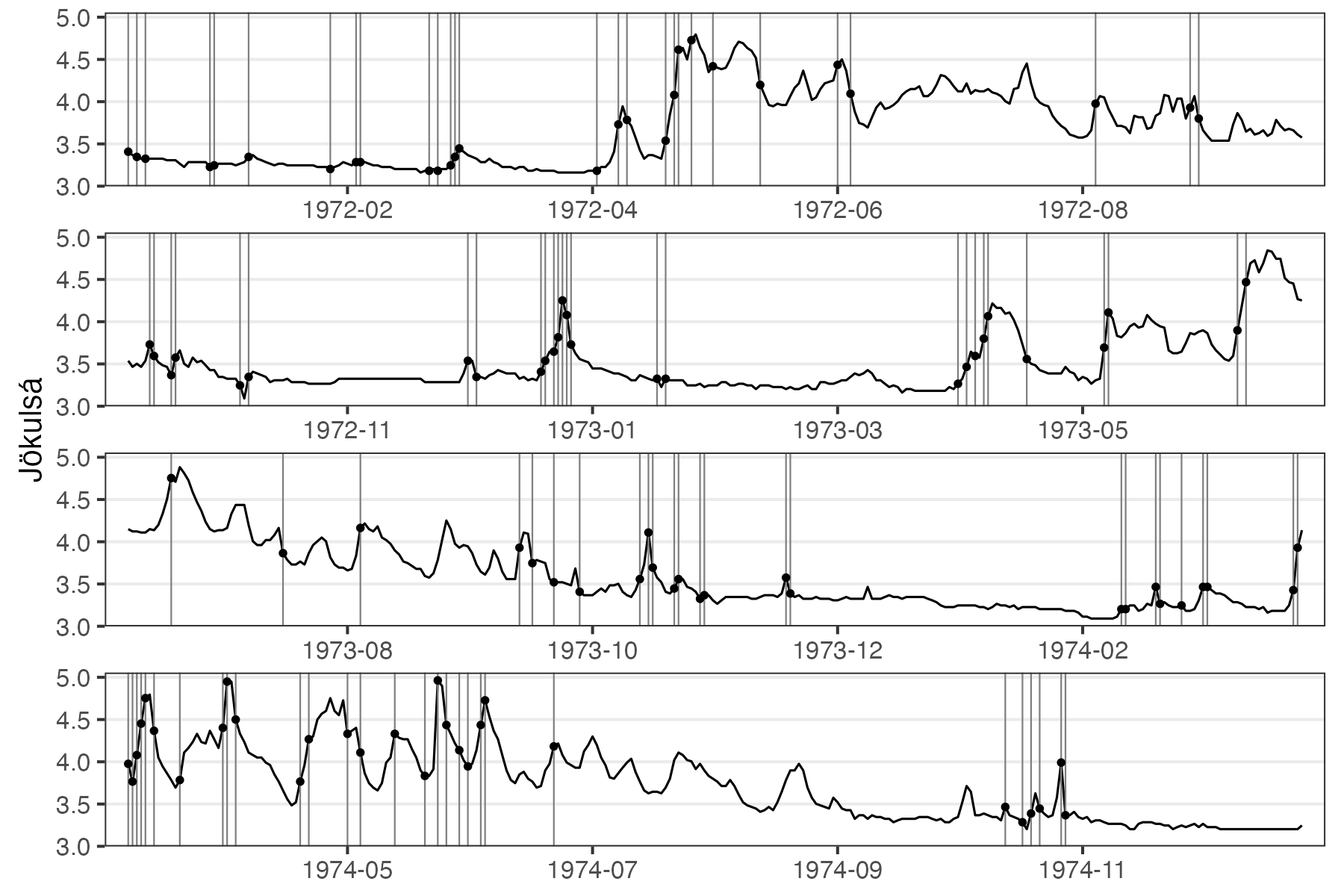}

  \vspace{0.5cm}

  \includegraphics[width=0.9\textwidth]{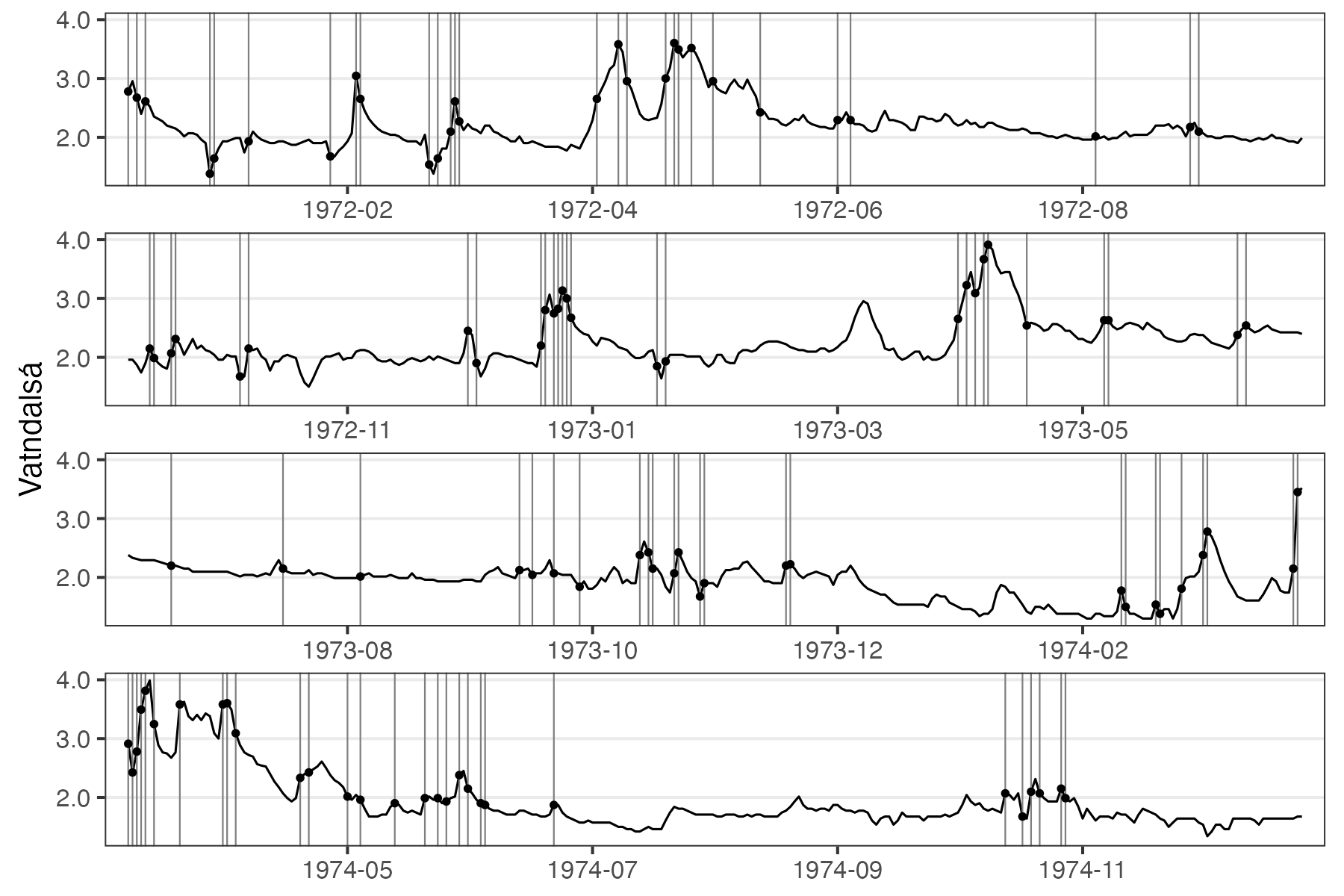}
  \label{fig:river_flags}
\end{figure}

For clarity, the full sample estimates are expressed in terms of the observed process $y_t$, whereas the robust estimates approximate the latent uncontaminated process $x_t$.

Based on Table~\ref{tab:aic_avg}, we adopt the VARMA$(2,0)$ specification. For this model, we set $\kappa=2$ in the iterative patch removal procedure and report estimates from the full sample, Huber-skip without patch removal, and patch removal.
The three estimators differ in whether they are based directly on the observed series $y_t$ or aim to recover the latent process $x_t$.
\begin{description}

\item[Full-sample estimates]
\begin{equation}
\begin{aligned}
z_{1, t} &= 1.233 z_{1, t-1} + 0.064 z_{2, t-1} - 0.258 z_{1, t-2} - 0.072 z_{2, t-2} + \varepsilon_{1, t}, \\
z_{2, t} &= 0.164 z_{1, t-1} + 0.993 z_{2, t-1} - 0.173 z_{1, t-2} - 0.064 z_{2, t-2} + \varepsilon_{2, t},
\end{aligned}
\end{equation}
where $z_{1, t} = y_{1, t} - 3.239$ and $z_{2, t} = y_{2, t} - 2.094$.

\item[Huber-skip without patch removal]
\begin{equation}
\begin{aligned}
z_{1, t} &= 1.206 z_{1, t-1} + 0.103 z_{2, t-1} - 0.222 z_{1, t-2} - 0.107 z_{2, t-2} + \varepsilon_{1, t}, \\
z_{2, t} &= -0.017 z_{1, t-1} + 1.210 z_{2, t-1} + 0.015 z_{1, t-2} - 0.258 z_{2, t-2} + \varepsilon_{2, t},
\end{aligned}
\end{equation}
where $z_{1, t} = x_{1, t} - 3.077$ and $z_{2, t} = x_{2, t} - 1.929$.

\item[Huber-skip with patch removal ($\kappa = 2$)]
\begin{equation}
\begin{aligned}
z_{1, t} &= 1.291 z_{1, t-1} + 0.072 z_{2, t-1} - 0.310 z_{1, t-2} - 0.070 z_{2, t-2} + \varepsilon_{1, t}, \\
z_{2, t} &= -0.008 z_{1, t-1} + 1.270 z_{2, t-1} + 0.010 z_{1, t-2} - 0.312 z_{2, t-2} + \varepsilon_{2, t},
\end{aligned}
\end{equation}
where $z_{1, t} = x_{1, t} - 3.544$ and $z_{2, t} = x_{2, t} - 1.933$.

\end{description}

The estimated coefficients differ non-negligibly across the three approaches. 
The Huber-skip estimator removes extreme observations through pointwise trimming, but residual propagation still distorts the estimated dependence structure. 
By contrast, the patch removal estimator accounts for this propagation mechanism and yields estimates more consistent with the latent uncontaminated process. 
The resulting differences are systematic rather than incidental and are consistent with the theoretical results on criterion distortion under dynamic contamination.

The observations identified by the iterative procedure are illustrated in Figure~\ref{fig:river_flags}. 
A comparison with Figure~\ref{fig:AO_IO_combined} suggests that the flagged observations are more consistent with IO-type contamination, although no formal classification is attempted. 
In particular, extreme observations are often followed by gradual adjustments or short-lived reversals, reflecting propagation through the AR structure.

\section{Conclusion}\label{sec:conclusion}

This paper establishes that subsample-based estimation, as commonly used in robust time series analysis, is generically inconsistent in dynamic settings when contamination propagates through the residual filter. 
The failure is structural and arises because residual propagation distorts the criterion defining the estimator itself, rather than merely contaminating a small number of observations. 
Consequently, removing contaminated observations does not in general recover the clean-data objective.

We introduce a propagation-compatible transformation of retained index sets, formalised through the patch removal operator. 
The operator provides a general correction for a broad class of residual-based subsample estimators without modifying their internal structure. 
Under suitable conditions, it leaves the uncontaminated estimator asymptotically unchanged while restoring consistency under contamination. 
The analysis also clarifies an important distinction across model classes.
For pure VAR models the required correction is finite, whereas moving-average components may generate propagation effects with infinite memory.

These results have direct implications for empirical work. 
When outliers are suspected, inference should not be based on standard estimators applied to the contaminated sample, or on subsample procedures that remove only the visibly aberrant observations. 
Instead, inference should be based on an estimator that remains consistent for the clean-data mechanism under contamination.
This is particularly important when the aim is to assess structural features of the clean data generating mechanism or model adequacy more generally.

The paper also points to several natural extensions, which fall into four directions.
First, at the level of estimation, an important extension is the estimation of the contamination magnitude $\zeta_t$. 
This problem is inherently conditional on the one studied here, as meaningful estimation of contamination effects requires consistent estimation of the clean-model parameters. 
Further extensions concern the patch removal mechanism, including the use of heterogeneous trimming lengths across contamination events or across observations. 
While the present analysis shows that a uniform bound on $\kappa$ suffices for consistency, relaxing this restriction may improve flexibility in practice.

Second, forecasting is a natural extension. 
In many applications the primary objective is to predict the clean process $x_t$.
Since contamination is typically sparse and does not reflect the underlying dynamics, forecasts should be based on a consistent estimate of the clean model, so as to capture the behaviour of the system in the absence of contamination.

Third, an important direction concerns inference. 
Outliers and residual propagation induce a structural distortion in estimation, under which standard hypothesis testing is conducted on estimators that are not consistent under contamination. 
Consequently, hypothesis tests may produce spurious statistical evidence,
including misleading conclusions on model parameters and on structural features of the underlying data-generating mechanism, thereby distorting the interpretation of fitted models.
By restoring consistency at the estimation stage, the present framework provides a basis for developing inference procedures that remain valid under contamination, including tests that account for residual propagation.

Finally, there are extensions related to model evaluation and selection. 
In particular, the development of subsample-based information criteria tailored to this framework remains an open problem. 
The average AIC used in the present paper serves only as a practical device, and a principled criterion for contaminated dynamic models with patch removal requires separate analysis.

These directions are important but raise distinct issues, and are therefore best studied separately. 
The present paper provides a foundation for such developments.

There are also clear limitations. 
First, the analysis is developed under settings in which contamination is sufficiently sparse. 
This is not merely a technical restriction. 
When contamination is frequent or persistent, it becomes difficult to distinguish external outliers from features that belong to the underlying structural mechanism. 
Second, patch removal necessarily reduces the effective sample size. 
Consequently, finite-sample convergence cannot be faster than in uncontaminated full sample estimation, and is slower when contamination is substantial or when the required patch length is large. 
This cost is unavoidable if consistency for the clean parameter is to be preserved.

For these reasons, the main practical recommendation of the paper is simple. 
If there are good reasons to suspect outliers in a dynamic time series, empirical conclusions should be based on an estimation procedure that is consistent for the clean model under contamination. 
The framework developed here provides a principled way to achieve this by correcting for residual propagation, thereby restoring the validity of subsample-based estimation in dynamic models.

\ifanonymous

\else
\subsection*{Acknowledgements}

We acknowledge financial support from the Jan Wallander and Tom Hedelius Foundation, Grant No.\ P2016-0293:1. 
Sandberg also acknowledges support from the same foundation, Grant No.\ P22-0264.
The computations were enabled by resources provided by the Swedish National Infrastructure for Computing (SNIC) at UPPMAX, partially funded by the Swedish Research Council under grant agreement No.\ 2018-05973.

\subsection*{Data Availability Statement}

The empirical dataset consists of daily Icelandic river flows originally analysed by \citet{ToThGu85} and widely used in the time series literature. The data are publicly available from the original source but cannot be redistributed due to licensing restrictions. All code required to reproduce the simulation and empirical results is available at \href{https://github.com/yukai-yang/Robust_Experiments}{github.com/yukai-yang/Robust\_Experiments} under the MIT license.
\fi


\bibliography{robust}
\bibliographystyle{chicago}


\appendix

\section{Proofs}

\begin{proof}[Proof of Proposition~\ref{prop:residual}]

The result follows directly by substituting the contamination model equations \eqref{eqn:additive_outlier} and \eqref{eqn:innovative_outlier} into the residual expressions.
\end{proof}

\begin{proof}[Proof of Theorem~\ref{thm:patch_invariance}]
Let
\begin{equation}
\Delta_T =
\sup_{\varphi\in\mathcal V}
\bigl|
f_T(\varphi,H_T)-f_T(\varphi,S^\kappa H_T)
\bigr|.
\end{equation}
By Assumption~\ref{ass:criterion_stability}, $\Delta_T \overset{p}{\longrightarrow}0$.
We consider the minimisation case. The maximisation case follows by symmetry.

\noindent
\textit{(i) Transferred local separation.}
By Assumption~\ref{ass:local_stab}(i), there exists $\delta>0$ such that, with probability tending to one,
\begin{equation}
\inf_{\varphi\notin\mathcal N} f_T(\varphi,H_T)
\ge
\inf_{\varphi\in\mathcal N} f_T(\varphi,H_T)+\delta.
\end{equation}
On the event $\{\Delta_T<\delta/4\}$, for any $\varphi\notin\mathcal N$ and $\psi\in\mathcal N$,
\begin{equation}
f_T(\varphi,S^\kappa H_T)
\ge
f_T(\varphi,H_T)-\Delta_T,
\qquad
f_T(\psi,S^\kappa H_T)
\le
f_T(\psi,H_T)+\Delta_T.
\end{equation}
Taking infima yields
\begin{equation}
\inf_{\varphi\notin\mathcal N} f_T(\varphi,S^\kappa H_T)
\ge
\inf_{\psi\in\mathcal N} f_T(\psi,S^\kappa H_T)+\delta-2\Delta_T
\ge
\inf_{\psi\in\mathcal N} f_T(\psi,S^\kappa H_T)+\delta/2.
\end{equation}
Thus (i) holds with $\delta'=\delta/2$.

By Assumption~\ref{ass:local_stab}(i) and the transferred local separation just established, both $\hat{\varphi}_T(H_T)$ and $\hat{\varphi}_T(S^\kappa H_T)$ belong to $\mathcal N$ with probability tending to one.

\noindent
\textit{Invariance of the estimator.}
By the definition of $\Delta_T$,
\begin{equation}\label{eq:deltat_bound}
f_T(\hat{\varphi}_T(S^\kappa H_T),H_T)
\le
f_T(\hat{\varphi}_T(H_T),H_T) + 2\Delta_T.
\end{equation}

Fix $\varepsilon>0$. By Assumption~\ref{ass:local_stab}(ii), there exists $\eta_{\varepsilon/2}>0$ such that
\begin{equation}
\inf_{\substack{\varphi\in\mathcal N:\\
\|\varphi-\hat{\varphi}_T(H_T)\|\ge\varepsilon/2}}
f_T(\varphi,H_T)
\ge
f_T(\hat{\varphi}_T(H_T),H_T)+\eta_{\varepsilon/2}.
\end{equation}
Therefore, on the event $\{2\Delta_T<\eta_{\varepsilon/2}\}$, the inequality
\begin{equation}
\|\hat{\varphi}_T(S^\kappa H_T)-\hat{\varphi}_T(H_T)\|\ge\varepsilon/2
\end{equation}
implies
\begin{equation}
f_T(\hat{\varphi}_T(S^\kappa H_T),H_T)
\ge
f_T(\hat{\varphi}_T(H_T),H_T)+\eta_{\varepsilon/2},
\end{equation}
which contradicts the previous bound \eqref{eq:deltat_bound}.
Hence
\begin{equation}
\|\hat{\varphi}_T(S^\kappa H_T)-\hat{\varphi}_T(H_T)\|
\overset{p}{\longrightarrow}0.
\end{equation}

\noindent
\textit{(ii) Transferred local sharpness.}
Fix $\varepsilon>0$. By the convergence just proved,
\begin{equation}
\|\hat{\varphi}_T(S^\kappa H_T)-\hat{\varphi}_T(H_T)\|
<\varepsilon/2
\end{equation}
with probability tending to one.

Now let $\varphi\in\mathcal N$ satisfy
$\|\varphi-\hat{\varphi}_T(S^\kappa H_T)\|\ge\varepsilon$.
Then, on the event above, the triangle inequality gives
$\|\varphi-\hat{\varphi}_T(H_T)\|
\ge
\varepsilon/2$.
Therefore, by Assumption~\ref{ass:local_stab}(ii),
\begin{equation}
f_T(\varphi,H_T)
\ge
f_T(\hat{\varphi}_T(H_T),H_T)+\eta_{\varepsilon/2}.
\end{equation}

Using the definition of $\Delta_T$ and the optimality of $\hat{\varphi}_T(S^\kappa H_T)$,
we obtain
\begin{equation}
f_T(\hat{\varphi}_T(H_T),H_T)
\ge
f_T(\hat{\varphi}_T(S^\kappa H_T),S^\kappa H_T)-\Delta_T,
\end{equation}
and hence
\begin{align}
f_T(\varphi,S^\kappa H_T)
&\ge
f_T(\varphi,H_T)-\Delta_T
\ge
f_T(\hat{\varphi}_T(H_T),H_T)+\eta_{\varepsilon/2}-\Delta_T \\
&\ge
f_T(\hat{\varphi}_T(S^\kappa H_T),S^\kappa H_T)
+\eta_{\varepsilon/2}-2\Delta_T.
\end{align}

On the event $\{\Delta_T<\eta_{\varepsilon/2}/4\}$, this yields
\begin{equation}
f_T(\varphi,S^\kappa H_T)
\ge
f_T(\hat{\varphi}_T(S^\kappa H_T),S^\kappa H_T)
+\eta_{\varepsilon/2}/2.
\end{equation}
Thus (ii) holds with $\eta'_\varepsilon=\eta_{\varepsilon/2}/2$.
Finally, by \eqref{eqn:PROperator2},
\begin{equation}
(K^\kappa\hat{\varphi}_T)(H_T)-\hat{\varphi}_T(H_T)
\overset{p}{\longrightarrow}0.
\end{equation}
\end{proof}

\begin{proof}[Proof of Corollary~\ref{cor:patch_consistency}]
The result follows from Theorem~\ref{thm:patch_invariance} and the assumption that $\hat{\varphi}_T(H_T)\overset{p}{\longrightarrow}\varphi_0$.
\end{proof}

Denote $\tilde{e}^T(\varphi, S^\kappa H_T)$ and $e^T(\varphi, S^\kappa H_T)$
vectors consisting of $\tilde{e}_{t}(\varphi)$ and $e_{t}(\varphi)$,
for $t \in S^\kappa H_T$, respectively.

\begin{lemma}\label{lem:residualconv_uniform}
Suppose that Assumptions~\ref{ass:outlier_rate} and \ref{ass:contamination_decay} hold, and that the process generated by \eqref{eqn:arma} is contaminated with AO or IO outliers as in \eqref{eqn:additive_outlier} or \eqref{eqn:innovative_outlier}.
Assume that
$\delta_t = 0$ for all $t \in H_T$.
Then
\begin{equation}\label{eq:evectorconvergence}
\sup_{\varphi \in \mathcal V}
\|\tilde{e}^T(\varphi, S^\kappa H_T)-e^T(\varphi, S^\kappa H_T)\|_\ell
\longrightarrow 0,
\end{equation}
for any $\ell \ge 1$ as $T\to\infty$.

In particular, in the VAR($p$) case, for both AO and IO contamination, if $\kappa = p < \infty$, then
\begin{equation}
\sup_{\varphi \in \mathcal V}
\|\tilde{e}^T(\varphi, S^\kappa H_T)-e^T(\varphi, S^\kappa H_T)\|_\ell
= 0.
\end{equation}
\end{lemma}

\begin{proof}
Since $\delta_t=0$ for all $t\in H_T$, every outlier lies in $H_T^c$, and patch removal excludes the $\kappa$ subsequent observations after each outlier. Hence only the exponentially decaying tail of the contamination remains on $S^\kappa H_T$.
All bounds below hold uniformly over $\varphi \in \mathcal V$, since the decay constants $M$ and $R$ can be chosen uniformly over $\mathcal V$.

We first consider the case of a single AO outlier. Suppose that an AO $\zeta_\tau$ occurs at time $\tau$. By \eqref{residualao_general}, for any $\varphi \in \mathcal V$,
\begin{equation}
\left| \tilde{e}_{\tau+k}(\varphi)-e_{\tau+k}(\varphi)\right|
=
|\pi_k|\,|\zeta_\tau|
\le
M R^k |\zeta_\tau|,
\qquad k=1,\ldots,T-\tau.
\end{equation}

Hence,
\begin{equation}
\sum_{k=1}^{T-\tau}
\left| \tilde{e}_{\tau+k}(\varphi)-e_{\tau+k}(\varphi)\right|
\le
M |\zeta_\tau| \sum_{k=1}^{T-\tau} R^k
\le
\frac{M |\zeta_\tau| R}{1-R}.
\end{equation}

After removing the observations $\{\tau,\tau+1,\ldots,\tau+\kappa\}$, the remaining contribution corresponds to the tail of the geometric sequence, so
\begin{equation}
\sum_{t \in S^\kappa H_T}
\left| \tilde{e}_t(\varphi)-e_t(\varphi)\right|
\le
M |\zeta_\tau| \sum_{k=\kappa+1}^{\infty} R^k
=
\frac{M |\zeta_\tau| R^{\kappa+1}}{1-R}.
\end{equation}

We now turn to the case of multiple AO outliers.
By Assumption~\ref{ass:outlier_rate}, if the trimming proportion $\alpha$ satisfies $\alpha>\bar\alpha$, then with probability tending to one,
$\sum_{t=1}^T \delta_t \le \alpha T$.
From \eqref{eqn:additive_outlier} and \eqref{residualao_general}, for any $t$,
\begin{equation}
\left| \tilde{e}_t(\varphi)-e_t(\varphi)\right|
\le
\sum_{j\ge 0} M R^j |\delta_{t-j}\zeta_{t-j}|
\le
\|\delta_T\zeta_T\|_\infty \sum_{j\ge 0} M R^j
=
\frac{M}{1-R}\|\delta_T\zeta_T\|_\infty.
\end{equation}

After removing $\kappa$ observations following each detected outlier, the remaining contribution of each outlier is reduced by a factor of order $R^\kappa$. More precisely, each outlier contributes at most a constant multiple of $\|\delta_T \zeta_T\|_\infty R^\kappa$ to the residual sum. Summing over at most $\alpha T$ outliers yields
\begin{equation}
\sup_{\varphi \in \mathcal V}
\sum_{t \in S^\kappa H_T}
\left| \tilde{e}_t(\varphi)-e_t(\varphi)\right|
\le
\frac{M}{1-R}\,\alpha T\,\|\delta_T\zeta_T\|_\infty\,R^\kappa.
\end{equation}

Therefore, under Assumption~\ref{ass:contamination_decay},
\begin{equation}
\sup_{\varphi \in \mathcal V}
\|\tilde{e}^T(\varphi,S^\kappa H_T)-e^T(\varphi,S^\kappa H_T)\|_1 \to 0.
\end{equation}

Finally, for any $\ell\ge 1$,
\begin{equation}
\sup_{\varphi \in \mathcal V}
\|\tilde{e}^T(\varphi,S^\kappa H_T)-e^T(\varphi,S^\kappa H_T)\|_\ell
\le
\sup_{\varphi \in \mathcal V}
\|\tilde{e}^T(\varphi,S^\kappa H_T)-e^T(\varphi,S^\kappa H_T)\|_1
\to 0.
\end{equation}

The same argument applies to IO contamination, since the corresponding propagation filter also has exponentially decaying coefficients under the causal and invertible VARMA assumptions.

In particular, in the VAR($p$) case, the propagation filter $\pi(L)$ has finite order $p$.
Hence,
for both AO and IO contamination, the residual effect of an outlier has finite length, with maximal lag at most $p$.
If $\kappa = p$, patch removal eliminates the entire contamination effect, so that
$\tilde{e}_t(\varphi)=e_t(\varphi)$
for all $t \in S^\kappa H_T$
uniformly over $\varphi \in \mathcal V$.
Consequently,
\begin{equation}
\sup_{\varphi \in \mathcal V}
\|\tilde{e}^T(\varphi,S^\kappa H_T)-e^T(\varphi,S^\kappa H_T)\|_\ell
= 0.
\end{equation}
\end{proof}

\begin{proof}[Proof of Theorem~\ref{thm:invariance}]
For notational simplicity, write
\begin{equation}
f_T^\kappa(\varphi)=f_T(\varphi,S^\kappa H_T),
\qquad
\tilde f_T^\kappa(\varphi)=\tilde f_T(\varphi,S^\kappa H_T),
\end{equation}
where, by Assumption~\ref{ass:criterion_lipschitz}(a),
\begin{equation}\label{eq:fTkfTk}
f_T^\kappa(\varphi)
=
\frac{1}{|S^\kappa H_T|}
\sum_{t\in S^\kappa H_T}
m_t\bigl(e_t(\varphi)\bigr),
\qquad
\tilde f_T^\kappa(\varphi)
=
\frac{1}{|S^\kappa H_T|}
\sum_{t\in S^\kappa H_T}
m_t\bigl(\tilde e_t(\varphi)\bigr).
\end{equation}
and let
$\varphi_T^\star=\hat{\varphi}_T(S^\kappa H_T)$,
$\tilde{\varphi}_T^\star=\tilde{\varphi}_T(S^\kappa H_T)$.

By Assumption~\ref{as:ultimate}, the event
$E_T=\{\delta_t=0 \text{ for all } t\in H_T\}$
satisfies $\Pr(E_T)\to 1$.
On the event $E_T$, Lemma~\ref{lem:residualconv_uniform} gives \eqref{eq:evectorconvergence}.
In the pure VAR($p$) case with $\kappa=p$, the same conclusion holds without Assumption~\ref{ass:contamination_decay}, since Lemma~\ref{lem:residualconv_uniform} then yields
\begin{equation}
\sup_{\varphi\in\mathcal V}
\|\tilde e^T(\varphi,S^\kappa H_T)-e^T(\varphi,S^\kappa H_T)\|_\ell
=0.
\end{equation}
\eqref{eq:fTkfTk} implies that
\begin{equation}
\bigl|
\tilde f_T^\kappa(\varphi)-f_T^\kappa(\varphi)
\bigr|
\le
\frac{1}{|S^\kappa H_T|}
\sum_{t\in S^\kappa H_T}
\bigl|
m_t(\tilde e_t(\varphi))-m_t(e_t(\varphi))
\bigr|.
\end{equation}
By Assumption~\ref{ass:criterion_lipschitz}(b), for each $t$,
\begin{equation}
\bigl|
m_t(\tilde e_t(\varphi))-m_t(e_t(\varphi))
\bigr|
\le
C\bigl(1+|e_t(\varphi)|+|\tilde e_t(\varphi)-e_t(\varphi)|\bigr)
|\tilde e_t(\varphi)-e_t(\varphi)|.
\end{equation}
Then it follows that
\begin{align}
\bigl|
\tilde f_T^\kappa(\varphi)-f_T^\kappa(\varphi)
\bigr|
&\le
\frac{C}{|S^\kappa H_T|}
\sum_{t\in S^\kappa H_T}
|\tilde e_t(\varphi)-e_t(\varphi)|
+
\frac{C}{|S^\kappa H_T|}
\sum_{t\in S^\kappa H_T}
|e_t(\varphi)|\,|\tilde e_t(\varphi)-e_t(\varphi)| \nonumber\\
&\quad+
\frac{C}{|S^\kappa H_T|}
\sum_{t\in S^\kappa H_T}
|\tilde e_t(\varphi)-e_t(\varphi)|^2.
\end{align}
The middle term is bounded by Cauchy--Schwarz,
\begin{equation}
\sum_{t\in S^\kappa H_T}
|e_t(\varphi)|\,|\tilde e_t(\varphi)-e_t(\varphi)|
\le
\left(
\sum_{t\in S^\kappa H_T}
|e_t(\varphi)|^2
\right)^{1/2}
\left(
\sum_{t\in S^\kappa H_T}
|\tilde e_t(\varphi)-e_t(\varphi)|^2
\right)^{1/2}.
\end{equation}
Taking suprema over $\varphi\in\mathcal V$, and using Lemma~\ref{lem:residualconv_uniform} together with Assumption~\ref{ass:criterion_lipschitz}(c), we obtain
\begin{equation}
\Delta_T
:=
\sup_{\varphi\in\mathcal V}
\bigl|
\tilde f_T^\kappa(\varphi)-f_T^\kappa(\varphi)
\bigr|
\overset{p}{\longrightarrow}0.
\end{equation}
By Theorem~\ref{thm:patch_invariance}, the criterion $f_T^\kappa$ satisfies local separation and sharpness on $\mathcal N$. In the case of minimisation, there exists $\delta>0$ such that, with probability tending to one,
\begin{equation}
\inf_{\varphi\notin\mathcal N} f_T^\kappa(\varphi)
\ge
\inf_{\varphi\in\mathcal N} f_T^\kappa(\varphi)+\delta.
\end{equation}
On the event $\{\Delta_T<\delta/4\}$, for any $\varphi\notin\mathcal N$ and $\psi\in\mathcal N$,
\begin{equation}
\tilde f_T^\kappa(\varphi)\ge f_T^\kappa(\varphi)-\Delta_T,
\qquad
\tilde f_T^\kappa(\psi)\le f_T^\kappa(\psi)+\Delta_T.
\end{equation}
Taking infima yields
\begin{equation}
\inf_{\varphi\notin\mathcal N} \tilde f_T^\kappa(\varphi)
\ge
\inf_{\psi\in\mathcal N} \tilde f_T^\kappa(\psi)+\delta-2\Delta_T
\ge
\inf_{\psi\in\mathcal N} \tilde f_T^\kappa(\psi)+\delta/2.
\end{equation}
Hence the minimiser of $\tilde f_T^\kappa$ belongs to $\mathcal N$ with probability tending to one. The maximisation case is analogous. Therefore, both $\varphi_T^\star$ and $\tilde{\varphi}_T^\star$ belong to $\mathcal N$ with probability tending to one.

Fix $\varepsilon>0$. By the local sharpness condition for $S^\kappa H_T$, there exists $\eta_\varepsilon>0$ such that, with probability tending to one,
\begin{equation}
\inf_{\substack{\varphi\in\mathcal N:\\
\|\varphi-\varphi_T^\star\|\ge\varepsilon}}
f_T^\kappa(\varphi)
\ge
f_T^\kappa(\varphi_T^\star)+\eta_\varepsilon,
\end{equation}
or, in the case of maximisation,
\begin{equation}
\sup_{\substack{\varphi\in\mathcal N:\\
\|\varphi-\varphi_T^\star\|\ge\varepsilon}}
f_T^\kappa(\varphi)
\le
f_T^\kappa(\varphi_T^\star)-\eta_\varepsilon.
\end{equation}

Consider the minimisation case. If
$\|\tilde{\varphi}_T^\star-\varphi_T^\star\|\ge\varepsilon$,
then, since both minimisers lie in $\mathcal N$ with probability tending to one,
\begin{equation}
f_T^\kappa(\tilde{\varphi}_T^\star)
\ge
f_T^\kappa(\varphi_T^\star)+\eta_\varepsilon.
\end{equation}
On the other hand, by optimality of $\tilde{\varphi}_T^\star$,
$\tilde f_T^\kappa(\tilde{\varphi}_T^\star)
\le
\tilde f_T^\kappa(\varphi_T^\star)$.
Combining this with the definition of $\Delta_T$, we obtain
\begin{equation}
f_T^\kappa(\tilde{\varphi}_T^\star)
\le
f_T^\kappa(\varphi_T^\star)+2\Delta_T.
\end{equation}
Hence,
$\eta_\varepsilon \le 2\Delta_T$,
which implies
$\Pr\bigl(\|\tilde{\varphi}_T^\star-\varphi_T^\star\|\ge\varepsilon\bigr)\longrightarrow 0$.
Thus,
$\tilde{\varphi}_T^\star-\varphi_T^\star\overset{p}{\longrightarrow}0$.
The argument for maximisation is analogous. Therefore,
\begin{equation}
\tilde{\varphi}_T(S^\kappa H_T)-\hat{\varphi}_T(S^\kappa H_T)\overset{p}{\longrightarrow}0.
\end{equation}
Since
\begin{equation}
(K^\kappa \tilde{\varphi}_T)(H_T)=\tilde{\varphi}_T(S^\kappa H_T),
\qquad
(K^\kappa \hat{\varphi}_T)(H_T)=\hat{\varphi}_T(S^\kappa H_T),
\end{equation}
it follows that
\begin{equation}
(K^\kappa \tilde{\varphi}_T)(H_T)-(K^\kappa \hat{\varphi}_T)(H_T)
\overset{p}{\longrightarrow}0.
\end{equation}
This completes the proof.
\end{proof}

\begin{proof}[Proof of Corollary~\ref{cor:consistency}]
By Corollary~\ref{cor:patch_consistency},
$\hat{\varphi}_T(S^\kappa H_T)\overset{p}{\longrightarrow}\varphi_0$,
since $\hat{\varphi}_T(H_T)\overset{p}{\longrightarrow}\varphi_0$ by assumption.
Combining this with Theorem~\ref{thm:invariance} yields
\begin{equation}
\tilde{\varphi}_T(S^\kappa H_T)\overset{p}{\longrightarrow}\varphi_0.
\end{equation}
\end{proof}

\begin{proof}[Proof of Proposition~\ref{prop:extremum_to_localstab}]

Part (a) follows from the standard consistency theorem for extremum estimators on a compact parameter space; see, for example, \citet{vdV1998}.
Indeed, by conditions (i)--(iii), the sample criterion is continuous, converges uniformly in probability to a continuous deterministic limit, and the limit function has a unique minimiser at $\varphi_0$. Therefore,
$\hat{\varphi}_T(H_T)\overset{p}{\longrightarrow}\varphi_0$.

It remains to verify Assumption~\ref{ass:local_stab}. Let $\mathcal N$ be any open neighbourhood of $\varphi_0$.
We first prove local separation. Since $\mathcal V$ is compact and $\mathcal N$ is open, the set
$\mathcal C:=\mathcal V\cap \mathcal N^c$
is compact. Because $f$ is continuous and $\varphi_0$ is the unique minimiser of $f$, we have
\begin{equation}
\delta_0
:=
\inf_{\varphi\in\mathcal C}
\bigl(f(\varphi)-f(\varphi_0)\bigr)
>0.
\end{equation}
Let $\delta=\delta_0/3$. On the event
\begin{equation}
\sup_{\varphi\in\mathcal V}
|f_T(\varphi,H_T)-f(\varphi)|
<
\delta,
\end{equation}
it follows that, for every $\varphi\notin\mathcal N$,
\begin{equation}
f_T(\varphi,H_T)
\ge
f(\varphi)-\delta
\ge
f(\varphi_0)+2\delta.
\end{equation}
On the other hand, since $\varphi_0\in\mathcal N$,
\begin{equation}
\inf_{\psi\in\mathcal N} f_T(\psi,H_T)
\le
f_T(\varphi_0,H_T)
\le
f(\varphi_0)+\delta.
\end{equation}
Hence,
\begin{equation}
\inf_{\varphi\notin\mathcal N} f_T(\varphi,H_T)
\ge
\inf_{\psi\in\mathcal N} f_T(\psi,H_T)+\delta.
\end{equation}
Since the event above has probability tending to one by condition (ii), local separation holds with probability tending to one.

We next prove local sharpness. Fix $\varepsilon>0$. Since $\mathcal N$ is an open neighbourhood of $\varphi_0$, there exists $r_\mathcal N>0$ such that
$B(\varphi_0,r_\mathcal N)\subset\mathcal N$.
Let
$r:=\min\{r_\mathcal N,\varepsilon/2\}$.
Then $r>0$ and
$B(\varphi_0,r)\subset\mathcal N$.

Now define
\begin{equation}
\eta_0
:=
\inf_{\substack{\varphi\in\mathcal V:\\
\|\varphi-\varphi_0\|\ge \varepsilon/2}}
\bigl(f(\varphi)-f(\varphi_0)\bigr).
\end{equation}
The set
$\Bigl\{\varphi\in\mathcal V:\|\varphi-\varphi_0\|\ge \varepsilon/2\Bigr\}$
is compact, since it is a closed subset of the compact set $\mathcal V$. Since $f$ is continuous and $\varphi_0$ is the unique minimiser of $f$, we have
$\eta_0>0$.
Let $\eta=\eta_0/3$.

Consider the event
$\|\hat{\varphi}_T(H_T)-\varphi_0\|<r$
and
\begin{equation}
\sup_{\varphi\in\mathcal V}
|f_T(\varphi,H_T)-f(\varphi)|
<
\eta.
\end{equation}
By part (a) and condition (ii), this event has probability tending to one.

On this event, we first note that
$\hat{\varphi}_T(H_T)\in B(\varphi_0,r)\subset\mathcal N$.
Now let $\varphi\in\mathcal N$ satisfy
$\|\varphi-\hat{\varphi}_T(H_T)\|\ge\varepsilon$.
Then, by the triangle inequality,
\begin{equation}
\|\varphi-\varphi_0\|
\ge
\|\varphi-\hat{\varphi}_T(H_T)\|
-
\|\hat{\varphi}_T(H_T)-\varphi_0\|
\ge
\varepsilon-r
\ge
\varepsilon/2.
\end{equation}
Therefore,
\begin{equation}
f(\varphi)\ge f(\varphi_0)+\eta_0=f(\varphi_0)+3\eta.
\end{equation}
Using the uniform approximation bound, we obtain
\begin{equation}
f_T(\varphi,H_T)
\ge
f(\varphi)-\eta
\ge
f(\varphi_0)+2\eta.
\end{equation}
Also,
\begin{equation}
f_T(\varphi_0,H_T)\le f(\varphi_0)+\eta.
\end{equation}
Since $\hat{\varphi}_T(H_T)$ minimises $f_T(\cdot,H_T)$ over $\mathcal V$,
\begin{equation}
f_T(\hat{\varphi}_T(H_T),H_T)\le f_T(\varphi_0,H_T)\le f(\varphi_0)+\eta.
\end{equation}
Combining the preceding inequalities gives
\begin{equation}
f_T(\varphi,H_T)
\ge
f_T(\hat{\varphi}_T(H_T),H_T)+\eta.
\end{equation}
Since this holds for every $\varphi\in\mathcal N$ such that
$\|\varphi-\hat{\varphi}_T(H_T)\|\ge\varepsilon$,
we conclude that
\begin{equation}
\inf_{\substack{\varphi\in\mathcal N:\\
\|\varphi-\hat{\varphi}_T(H_T)\|\ge\varepsilon}}
f_T(\varphi,H_T)
\ge
f_T(\hat{\varphi}_T(H_T),H_T)+\eta.
\end{equation}
Thus local sharpness holds with probability tending to one.

We have therefore shown that both local separation and local sharpness hold with probability tending to one. Hence Assumption~\ref{ass:local_stab} is satisfied.
\end{proof}

\begin{proof}[Proof of Proposition~\ref{prop:patch_stability}]

Write
$h_T = |H_T|$,
$h_T^\kappa = |S^\kappa H_T|$,
$r_T = |H_T\setminus S^\kappa H_T|$.
Then $h_T^\kappa = h_T - r_T$.
For any $\varphi\in\mathcal V$, we decompose
\begin{align}
\bigl|f_T(\varphi,H_T)-f_T(\varphi,S^\kappa H_T)\bigr|
&=
\left|
\frac{1}{h_T}\sum_{t\in H_T} m_t(e_t(\varphi))
-
\frac{1}{h_T^\kappa}\sum_{t\in S^\kappa H_T} m_t(e_t(\varphi))
\right| \nonumber \\
&\le
\left|
\frac{1}{h_T}\sum_{t\in H_T\setminus S^\kappa H_T} m_t(e_t(\varphi))
\right|
+
\left|
\left(
\frac{1}{h_T}-\frac{1}{h_T^\kappa}
\right)
\sum_{t\in S^\kappa H_T} m_t(e_t(\varphi))
\right|.
\end{align}

The first term is bounded by
\begin{equation}
\frac{r_T}{h_T}
\sup_{t\in H_T} |m_t(e_t(\varphi))|.
\end{equation}

For the second term, note that
\begin{equation}
\left|
\frac{1}{h_T}-\frac{1}{h_T^\kappa}
\right|
=
\frac{r_T}{h_T\,h_T^\kappa}.
\end{equation}
Hence
\begin{align}
\left|
\left(
\frac{1}{h_T}-\frac{1}{h_T^\kappa}
\right)
\sum_{t\in S^\kappa H_T} m_t(e_t(\varphi))
\right|
&\le
\frac{r_T}{h_T\,h_T^\kappa}
\cdot
h_T^\kappa
\sup_{t\in H_T} |m_t(e_t(\varphi))| \\
&=
\frac{r_T}{h_T}
\sup_{t\in H_T} |m_t(e_t(\varphi))|.
\end{align}
Combining the two bounds yields
\begin{equation}
\bigl|f_T(\varphi,H_T)-f_T(\varphi,S^\kappa H_T)\bigr|
\le
2\frac{r_T}{h_T}
\sup_{t\in H_T} |m_t(e_t(\varphi))|.
\end{equation}
Taking supremum over $\varphi\in\mathcal V$, we obtain
\begin{equation}
\sup_{\varphi\in\mathcal V}
\bigl|f_T(\varphi,H_T)-f_T(\varphi,S^\kappa H_T)\bigr|
\le
2\frac{r_T}{h_T}
\sup_{\varphi\in\mathcal V}\sup_{t\in H_T}
|m_t(e_t(\varphi))|.
\end{equation}
By assumptions (i) and (ii),
\begin{equation}
\frac{r_T}{h_T} \overset{p}{\longrightarrow} 0,
\qquad
\sup_{\varphi,t}|m_t(e_t(\varphi))| = O_p(1),
\end{equation}
and therefore their product converges to zero in probability. Hence
\begin{equation}
\sup_{\varphi\in\mathcal V}
\bigl|f_T(\varphi,H_T)-f_T(\varphi,S^\kappa H_T)\bigr|
\overset{p}{\longrightarrow}0,
\end{equation}
which proves Assumption~\ref{ass:criterion_stability}.
\end{proof}

\section{Tables}

\begin{table}[!htbp]
\centering
\caption{Total bias and RMSE for the VAR model. “Clean” denotes uncontaminated data. For the VAR model, IO yields results identical to the clean case.}
\footnotesize
\renewcommand{\arraystretch}{1.1}
\begin{tabular}{cc|cccc|cccc}
\hline
& & \multicolumn{4}{c|}{$T=500$} & \multicolumn{4}{c}{$T=1000$} \\
$\zeta$ & $\alpha(\%)$
& \multicolumn{2}{c}{Clean/IO} & \multicolumn{2}{c|}{AO}
& \multicolumn{2}{c}{Clean/IO} & \multicolumn{2}{c}{AO} \\
& & $\kappa=0$ & $\kappa=1$ & $\kappa=0$ & $\kappa=1$
& $\kappa=0$ & $\kappa=1$ & $\kappa=0$ & $\kappa=1$ \\
\hline
\multicolumn{10}{l}{\textit{Panel A: Total bias}}\\
5   & 1  & 0.0049 & 0.0049 & 0.1239 & 0.0049 & 0.0027 & 0.0028 & 0.1214 & 0.0028 \\
    & 5  & 0.0062 & 0.0060 & 0.4112 & 0.0060 & 0.0031 & 0.0035 & 0.4080 & 0.0035 \\
    & 10 & 0.0055 & 0.0052 & 0.5859 & 0.0052 & 0.0035 & 0.0035 & 0.5845 & 0.0035 \\
10  & 1  & 0.0056 & 0.0055 & 0.3517 & 0.0055 & 0.0035 & 0.0035 & 0.3500 & 0.0035 \\
    & 5  & 0.0036 & 0.0036 & 0.7125 & 0.0036 & 0.0030 & 0.0030 & 0.7152 & 0.0030 \\
    & 10 & 0.0059 & 0.0069 & 0.8198 & 0.0069 & 0.0032 & 0.0036 & 0.8217 & 0.0036 \\
50  & 1  & 0.0037 & 0.0036 & 0.8890 & 0.0036 & 0.0031 & 0.0031 & 0.8886 & 0.0031 \\
    & 5  & 0.0043 & 0.0043 & 0.9374 & 0.0043 & 0.0030 & 0.0027 & 0.9396 & 0.0027 \\
    & 10 & 0.0073 & 0.0075 & 0.9473 & 0.0075 & 0.0019 & 0.0023 & 0.9456 & 0.0023 \\
100 & 1  & 0.0054 & 0.0055 & 0.9375 & 0.0055 & 0.0043 & 0.0043 & 0.9380 & 0.0043 \\
    & 5  & 0.0054 & 0.0054 & 0.9430 & 0.0054 & 0.0031 & 0.0033 & 0.9479 & 0.0033 \\
    & 10 & 0.0048 & 0.0046 & 0.9437 & 0.0046 & 0.0017 & 0.0015 & 0.9483 & 0.0015 \\
\hline
\multicolumn{10}{l}{\textit{Panel B: RMSE}}\\
5   & 1  & 0.0653 & 0.0656 & 0.1445 & 0.0656 & 0.0462 & 0.0464 & 0.1323 & 0.0464 \\
    & 5  & 0.0686 & 0.0708 & 0.4219 & 0.0708 & 0.0468 & 0.0483 & 0.4132 & 0.0483 \\
    & 10 & 0.0696 & 0.0735 & 0.5949 & 0.0735 & 0.0487 & 0.0517 & 0.5889 & 0.0517 \\
10  & 1  & 0.0665 & 0.0669 & 0.3626 & 0.0669 & 0.0460 & 0.0461 & 0.3555 & 0.0461 \\
    & 5  & 0.0679 & 0.0691 & 0.7209 & 0.0691 & 0.0467 & 0.0479 & 0.7194 & 0.0479 \\
    & 10 & 0.0696 & 0.0743 & 0.8281 & 0.0743 & 0.0483 & 0.0518 & 0.8259 & 0.0518 \\
50  & 1  & 0.0645 & 0.0649 & 0.8974 & 0.0649 & 0.0460 & 0.0463 & 0.8927 & 0.0463 \\
    & 5  & 0.0663 & 0.0684 & 0.9456 & 0.0684 & 0.0472 & 0.0484 & 0.9440 & 0.0484 \\
    & 10 & 0.0694 & 0.0733 & 0.9559 & 0.0733 & 0.0484 & 0.0518 & 0.9497 & 0.0518 \\
100 & 1  & 0.0646 & 0.0649 & 0.9456 & 0.0649 & 0.0466 & 0.0467 & 0.9420 & 0.0467 \\
    & 5  & 0.0675 & 0.0694 & 0.9518 & 0.0694 & 0.0475 & 0.0487 & 0.9528 & 0.0487 \\
    & 10 & 0.0668 & 0.0709 & 0.9559 & 0.0709 & 0.0492 & 0.0518 & 0.9554 & 0.0518 \\
\hline
\end{tabular}
\label{tab:VAR_bias_RMSE}
\end{table}

\begin{sidewaystable}[ht]
\centering
\caption{Total bias and RMSE for the VMA model. “Clean” denotes uncontaminated data. For the VMA model, AO and IO yield identical results.}
\small
\renewcommand{\arraystretch}{1.2}
\resizebox{\textwidth}{!}{%
\begin{tabular}{cc|cccc cccc|cccc cccc}
\hline
& & \multicolumn{8}{c|}{$T=500$} & \multicolumn{8}{c}{$T=1000$} \\
$\zeta$ & $\alpha(\%)$
& \multicolumn{4}{c}{Clean} & \multicolumn{4}{c|}{AO/IO}
& \multicolumn{4}{c}{Clean} & \multicolumn{4}{c}{AO/IO} \\
& & $\kappa=0$ & $\kappa=2$ & $\kappa=5$ & $\kappa=9$
  & $\kappa=0$ & $\kappa=2$ & $\kappa=5$ & $\kappa=9$
  & $\kappa=0$ & $\kappa=2$ & $\kappa=5$ & $\kappa=9$
  & $\kappa=0$ & $\kappa=2$ & $\kappa=5$ & $\kappa=9$ \\
\hline
\multicolumn{18}{l}{\textit{Panel A: Total bias}}\\
5   & 1  & 0.0048 & 0.0048 & 0.0051 & 0.0052 & 0.2219 & 0.1384 & 0.0610 & 0.0182 & 0.0019 & 0.0019 & 0.0019 & 0.0020 & 0.2215 & 0.1387 & 0.0617 & 0.0176 \\
    & 5  & 0.0047 & 0.0047 & 0.0047 & 0.0047 & 0.3600 & 0.2097 & 0.0966 & 0.0319 & 0.0014 & 0.0013 & 0.0013 & 0.0013 & 0.3590 & 0.2103 & 0.0971 & 0.0327 \\
    & 10 & 0.0043 & 0.0042 & 0.0046 & 0.0058 & 0.6309 & 0.3488 & 0.1787 & 0.0741 & 0.0027 & 0.0026 & 0.0025 & 0.0032 & 0.6307 & 0.3493 & 0.1794 & 0.0736 \\
10  & 1  & 0.0047 & 0.0048 & 0.0049 & 0.0052 & 0.4367 & 0.2453 & 0.1141 & 0.0382 & 0.0030 & 0.0030 & 0.0031 & 0.0030 & 0.4363 & 0.2447 & 0.1140 & 0.0385 \\
    & 5  & 0.0036 & 0.0039 & 0.0041 & 0.0047 & 0.5908 & 0.3217 & 0.1569 & 0.0575 & 0.0023 & 0.0023 & 0.0023 & 0.0026 & 0.5906 & 0.3212 & 0.1566 & 0.0571 \\
    & 10 & 0.0043 & 0.0043 & 0.0052 & 0.0071 & 0.7813 & 0.4510 & 0.2427 & 0.1037 & 0.0029 & 0.0030 & 0.0029 & 0.0039 & 0.7801 & 0.4507 & 0.2432 & 0.1042 \\
50  & 1  & 0.0061 & 0.0061 & 0.0061 & 0.0060 & 0.8168 & 0.4867 & 0.2598 & 0.1047 & 0.0026 & 0.0026 & 0.0026 & 0.0028 & 0.8167 & 0.4871 & 0.2610 & 0.1054 \\
    & 5  & 0.0037 & 0.0039 & 0.0040 & 0.0042 & 0.8435 & 0.5420 & 0.3030 & 0.1301 & 0.0024 & 0.0023 & 0.0024 & 0.0027 & 0.8437 & 0.5427 & 0.3044 & 0.1298 \\
    & 10 & 0.0042 & 0.0044 & 0.0054 & 0.0060 & 0.8922 & 0.6280 & 0.3831 & 0.1930 & 0.0035 & 0.0038 & 0.0044 & 0.0043 & 0.9225 & 0.6284 & 0.3877 & 0.1935 \\
100 & 1  & 0.0061 & 0.0060 & 0.0061 & 0.0060 & 0.8519 & 0.5676 & 0.3233 & 0.1409 & 0.0026 & 0.0025 & 0.0025 & 0.0024 & 0.8510 & 0.5676 & 0.3239 & 0.1409 \\
    & 5  & 0.0050 & 0.0048 & 0.0054 & 0.0057 & 0.9125 & 0.6132 & 0.3645 & 0.1690 & 0.0023 & 0.0025 & 0.0024 & 0.0024 & 0.9253 & 0.6128 & 0.3667 & 0.1695 \\
    & 10 & 0.0049 & 0.0048 & 0.0052 & 0.0057 & 1.0449 & 0.6917 & 0.4392 & 0.2364 & 0.0033 & 0.0035 & 0.0036 & 0.0042 & 1.1100 & 0.7044 & 0.4430 & 0.2365 \\
\hline
\multicolumn{18}{l}{\textit{Panel B: RMSE}}\\
5   & 1  & 0.0614 & 0.0616 & 0.0620 & 0.0624 & 0.2315 & 0.1513 & 0.0848 & 0.0623 & 0.0426 & 0.0428 & 0.0430 & 0.0435 & 0.2262 & 0.1451 & 0.0742 & 0.0456 \\
    & 5  & 0.0632 & 0.0636 & 0.0646 & 0.0663 & 0.3671 & 0.2195 & 0.1150 & 0.0708 & 0.0447 & 0.0451 & 0.0456 & 0.0471 & 0.3625 & 0.2154 & 0.1065 & 0.0553 \\
    & 10 & 0.0629 & 0.0671 & 0.0748 & 0.0880 & 0.6357 & 0.3569 & 0.1935 & 0.1091 & 0.0452 & 0.0478 & 0.0521 & 0.0621 & 0.6333 & 0.3533 & 0.1866 & 0.0930 \\
10  & 1  & 0.0615 & 0.0616 & 0.0620 & 0.0626 & 0.4426 & 0.2534 & 0.1288 & 0.0699 & 0.0436 & 0.0437 & 0.0440 & 0.0444 & 0.4391 & 0.2486 & 0.1215 & 0.0568 \\
    & 5  & 0.0608 & 0.0616 & 0.0627 & 0.0644 & 0.5956 & 0.3285 & 0.1685 & 0.0826 & 0.0440 & 0.0444 & 0.0449 & 0.0460 & 0.5930 & 0.3248 & 0.1625 & 0.0709 \\
    & 10 & 0.0639 & 0.0670 & 0.0739 & 0.0901 & 0.7858 & 0.4584 & 0.2555 & 0.1304 & 0.0442 & 0.0471 & 0.0517 & 0.0620 & 0.7822 & 0.4542 & 0.2493 & 0.1171 \\
50  & 1  & 0.0628 & 0.0630 & 0.0634 & 0.0640 & 0.8211 & 0.4928 & 0.2688 & 0.1182 & 0.0445 & 0.0447 & 0.0448 & 0.0452 & 0.8186 & 0.4899 & 0.2654 & 0.1120 \\
    & 5  & 0.0624 & 0.0627 & 0.0634 & 0.0650 & 0.8635 & 0.5482 & 0.3121 & 0.1411 & 0.0438 & 0.0442 & 0.0450 & 0.0460 & 0.8504 & 0.5455 & 0.3091 & 0.1356 \\
    & 10 & 0.0649 & 0.0679 & 0.0744 & 0.0893 & 1.0320 & 0.6346 & 0.3981 & 0.2091 & 0.0444 & 0.0469 & 0.0521 & 0.0621 & 1.1003 & 0.6319 & 0.3954 & 0.2011 \\
100 & 1  & 0.0632 & 0.0635 & 0.0639 & 0.0645 & 0.8725 & 0.5733 & 0.3329 & 0.1509 & 0.0435 & 0.0437 & 0.0440 & 0.0444 & 0.8588 & 0.5702 & 0.3285 & 0.1458 \\
    & 5  & 0.0627 & 0.0631 & 0.0641 & 0.0655 & 1.0535 & 0.6186 & 0.3753 & 0.1786 & 0.0431 & 0.0436 & 0.0444 & 0.0456 & 1.0482 & 0.6155 & 0.3722 & 0.1740 \\
    & 10 & 0.0642 & 0.0672 & 0.0740 & 0.0896 & 1.5215 & 0.7598 & 0.4587 & 0.2528 & 0.0453 & 0.0479 & 0.0530 & 0.0622 & 1.6501 & 0.7838 & 0.4559 & 0.2445 \\
\hline
\end{tabular}}
\label{tab:VMA_bias_RMSE}
\end{sidewaystable}

\begin{table}[ht]
\centering
\caption{Total bias and RMSE for the VARMA model with $T=500$. 
“Clean” denotes uncontaminated data.}
\small
\renewcommand{\arraystretch}{1.2}
\setlength{\tabcolsep}{4pt}
\resizebox{\textwidth}{!}{%
\begin{tabular}{cc|cccc|cccc|cccc}
\hline
& & \multicolumn{4}{c|}{Clean} & \multicolumn{4}{c|}{AO} & \multicolumn{4}{c}{IO} \\
$\zeta$ & $\alpha(\%)$
& $\kappa=0$ & $\kappa=2$ & $\kappa=5$ & $\kappa=9$
& $\kappa=0$ & $\kappa=2$ & $\kappa=5$ & $\kappa=9$
& $\kappa=0$ & $\kappa=2$ & $\kappa=5$ & $\kappa=9$ \\
\hline
\multicolumn{14}{l}{\textit{Panel A: Total bias}}\\
5   & 1  & 0.0088 & 0.0090 & 0.0086 & 0.0088 & 0.4372 & 0.2677 & 0.1200 & 0.0401 & 0.2137 & 0.1350 & 0.0590 & 0.0186 \\
    & 5  & 0.0091 & 0.0087 & 0.0089 & 0.0088 & 0.6348 & 0.3575 & 0.1658 & 0.0602 & 0.3583 & 0.2118 & 0.0969 & 0.0310 \\
    & 10 & 0.0070 & 0.0067 & 0.0074 & 0.0090 & 1.0163 & 0.5098 & 0.2675 & 0.1139 & 0.6226 & 0.3493 & 0.1840 & 0.0778 \\
10  & 1  & 0.0081 & 0.0076 & 0.0085 & 0.0077 & 0.7387 & 0.3974 & 0.1888 & 0.0686 & 0.4333 & 0.2468 & 0.1155 & 0.0380 \\
    & 5  & 0.0052 & 0.0053 & 0.0058 & 0.0066 & 0.9561 & 0.4807 & 0.2423 & 0.0929 & 0.5834 & 0.3248 & 0.1615 & 0.0602 \\
    & 10 & 0.0079 & 0.0079 & 0.0098 & 0.0113 & 1.1901 & 0.6136 & 0.3385 & 0.1534 & 0.7641 & 0.4474 & 0.2517 & 0.1105 \\
50  & 1  & 0.0081 & 0.0081 & 0.0079 & 0.0084 & 1.2186 & 0.6428 & 0.3571 & 0.1535 & 0.8146 & 0.4818 & 0.2675 & 0.1097 \\
    & 5  & 0.0061 & 0.0062 & 0.0064 & 0.0067 & 1.1168 & 0.6914 & 0.4146 & 0.1866 & 0.9628 & 0.5420 & 0.3025 & 0.1379 \\
    & 10 & 0.0088 & 0.0088 & 0.0101 & 0.0124 & 1.0660 & 0.7751 & 0.4814 & 0.2635 & 1.5779 & 0.6804 & 0.3858 & 0.1975 \\
100 & 1  & 0.0103 & 0.0109 & 0.0112 & 0.0106 & 1.2988 & 0.7161 & 0.4258 & 0.2011 & 1.0811 & 0.5739 & 0.3212 & 0.1480 \\
    & 5  & 0.0071 & 0.0070 & 0.0077 & 0.0084 & 1.0479 & 0.7637 & 0.4683 & 0.2358 & 1.4761 & 0.6476 & 0.3638 & 0.1749 \\
    & 10 & 0.0072 & 0.0069 & 0.0075 & 0.0087 & 1.3776 & 0.8554 & 0.5198 & 0.3070 & 0.9864 & 0.6250 & 0.4479 & 0.2444 \\
\hline
\multicolumn{14}{l}{\textit{Panel B: RMSE}}\\
5   & 1  & 0.0956 & 0.0964 & 0.0963 & 0.0975 & 0.4577 & 0.2901 & 0.1571 & 0.1044 & 0.2429 & 0.1717 & 0.1154 & 0.0982 \\
    & 5  & 0.1010 & 0.1021 & 0.1039 & 0.1056 & 0.6676 & 0.3816 & 0.2018 & 0.1214 & 0.3799 & 0.2405 & 0.1453 & 0.1099 \\
    & 10 & 0.0997 & 0.1053 & 0.1166 & 0.1399 & 1.1598 & 0.5576 & 0.3020 & 0.1781 & 0.6566 & 0.3831 & 0.2246 & 0.1574 \\
10  & 1  & 0.0986 & 0.0986 & 0.0986 & 0.1002 & 0.7856 & 0.4238 & 0.2196 & 0.1210 & 0.4541 & 0.2731 & 0.1557 & 0.1060 \\
    & 5  & 0.0979 & 0.0988 & 0.0996 & 0.1020 & 1.0712 & 0.5191 & 0.2688 & 0.1360 & 0.6127 & 0.3492 & 0.1953 & 0.1174 \\
    & 10 & 0.1010 & 0.1072 & 0.1178 & 0.1417 & 1.5580 & 0.6700 & 0.3793 & 0.2068 & 0.8669 & 0.5082 & 0.2867 & 0.1756 \\
50  & 1  & 0.0978 & 0.0969 & 0.0976 & 0.0993 & 1.9750 & 0.6851 & 0.3886 & 0.1826 & 0.9994 & 0.5581 & 0.2946 & 0.1455 \\
    & 5  & 0.0977 & 0.0989 & 0.0995 & 0.1019 & 2.8339 & 0.8041 & 0.4497 & 0.2127 & 1.1198 & 0.6419 & 0.3323 & 0.1690 \\
    & 10 & 0.1011 & 0.1050 & 0.1162 & 0.1390 & 2.8165 & 0.9428 & 0.5334 & 0.3049 & 2.1401 & 0.8004 & 0.4193 & 0.2361 \\
100 & 1  & 0.0984 & 0.0995 & 0.0992 & 0.1007 & 3.6456 & 0.8353 & 0.4695 & 0.2256 & 1.3862 & 0.7206 & 0.3544 & 0.1761 \\
    & 5  & 0.0977 & 0.0989 & 0.1004 & 0.1011 & 2.6143 & 0.9281 & 0.5148 & 0.2617 & 2.0412 & 0.7878 & 0.3954 & 0.1992 \\
    & 10 & 0.0983 & 0.1035 & 0.1145 & 0.1413 & 2.6570 & 0.9886 & 0.6042 & 0.3506 & 1.6228 & 1.0168 & 0.5682 & 0.2800 \\
\hline
\end{tabular}}
\label{tab:VARMA_bias_RMSE_500}
\end{table}

\begin{table}[ht]
\centering
\caption{Total bias and RMSE for the VARMA model with $T=1000$. 
“Clean” denotes uncontaminated data.}
\small
\renewcommand{\arraystretch}{1.2}
\setlength{\tabcolsep}{4pt}
\resizebox{\textwidth}{!}{%
\begin{tabular}{cc|cccc|cccc|cccc}
\hline
& & \multicolumn{4}{c|}{Clean} & \multicolumn{4}{c|}{AO} & \multicolumn{4}{c}{IO} \\
$\zeta$ & $\alpha(\%)$
& $\kappa=0$ & $\kappa=2$ & $\kappa=5$ & $\kappa=9$
& $\kappa=0$ & $\kappa=2$ & $\kappa=5$ & $\kappa=9$
& $\kappa=0$ & $\kappa=2$ & $\kappa=5$ & $\kappa=9$ \\
\hline
\multicolumn{14}{l}{\textit{Panel A: Total bias}}\\
5   & 1  & 0.0045 & 0.0040 & 0.0043 & 0.0046 & 0.4374 & 0.2678 & 0.1213 & 0.0390 & 0.2135 & 0.1358 & 0.0595 & 0.0170 \\
    & 5  & 0.0039 & 0.0036 & 0.0042 & 0.0043 & 0.6328 & 0.3586 & 0.1667 & 0.0620 & 0.3564 & 0.2120 & 0.0990 & 0.0325 \\
    & 10 & 0.0048 & 0.0050 & 0.0055 & 0.0059 & 1.0177 & 0.5141 & 0.2690 & 0.1123 & 0.6185 & 0.3510 & 0.1848 & 0.0766 \\
10  & 1  & 0.0054 & 0.0054 & 0.0054 & 0.0054 & 0.7380 & 0.3975 & 0.1896 & 0.0693 & 0.4338 & 0.2456 & 0.1162 & 0.0385 \\
    & 5  & 0.0045 & 0.0047 & 0.0047 & 0.0057 & 0.9529 & 0.4806 & 0.2431 & 0.0913 & 0.5837 & 0.3243 & 0.1594 & 0.0589 \\
    & 10 & 0.0051 & 0.0054 & 0.0054 & 0.0071 & 1.2082 & 0.6109 & 0.3382 & 0.1534 & 0.7595 & 0.4483 & 0.2503 & 0.1106 \\
50  & 1  & 0.0045 & 0.0045 & 0.0045 & 0.0048 & 1.2105 & 0.6443 & 0.3573 & 0.1545 & 0.8145 & 0.4868 & 0.2694 & 0.1112 \\
    & 5  & 0.0046 & 0.0045 & 0.0045 & 0.0051 & 1.2570 & 0.6907 & 0.4206 & 0.1858 & 0.9542 & 0.5431 & 0.3038 & 0.1370 \\
    & 10 & 0.0052 & 0.0050 & 0.0050 & 0.0067 & 1.1624 & 0.7753 & 0.4846 & 0.2682 & 1.6098 & 0.6915 & 0.3917 & 0.1986 \\
100 & 1  & 0.0055 & 0.0054 & 0.0057 & 0.0053 & 1.4009 & 0.7150 & 0.4309 & 0.2014 & 1.0915 & 0.5994 & 0.3211 & 0.1492 \\
    & 5  & 0.0049 & 0.0049 & 0.0051 & 0.0050 & 1.0859 & 0.7620 & 0.4673 & 0.2355 & 1.5329 & 0.6388 & 0.3640 & 0.1752 \\
    & 10 & 0.0042 & 0.0045 & 0.0047 & 0.0059 & 1.3188 & 0.8570 & 0.5233 & 0.3086 & 1.0199 & 0.6414 & 0.4735 & 0.2455 \\
\hline
\multicolumn{14}{l}{\textit{Panel B: RMSE}}\\
5   & 1  & 0.0693 & 0.0695 & 0.0699 & 0.0700 & 0.4517 & 0.2816 & 0.1430 & 0.0809 & 0.2305 & 0.1580 & 0.0949 & 0.0709 \\
    & 5  & 0.0689 & 0.0695 & 0.0716 & 0.0734 & 0.6635 & 0.3728 & 0.1870 & 0.0958 & 0.3695 & 0.2270 & 0.1246 & 0.0797 \\
    & 10 & 0.0715 & 0.0741 & 0.0830 & 0.0983 & 1.1362 & 0.5513 & 0.2908 & 0.1486 & 0.6500 & 0.3720 & 0.2081 & 0.1240 \\
10  & 1  & 0.0687 & 0.0691 & 0.0693 & 0.0699 & 0.7778 & 0.4138 & 0.2073 & 0.0984 & 0.4465 & 0.2606 & 0.1380 & 0.0802 \\
    & 5  & 0.0682 & 0.0684 & 0.0697 & 0.0714 & 1.0734 & 0.5100 & 0.2591 & 0.1151 & 0.6069 & 0.3386 & 0.1807 & 0.0927 \\
    & 10 & 0.0702 & 0.0742 & 0.0811 & 0.0981 & 1.5696 & 0.6610 & 0.3624 & 0.1821 & 0.8666 & 0.5013 & 0.2714 & 0.1461 \\
50  & 1  & 0.0692 & 0.0702 & 0.0700 & 0.0698 & 2.0872 & 0.6711 & 0.3797 & 0.1706 & 0.9814 & 0.5523 & 0.2889 & 0.1303 \\
    & 5  & 0.0684 & 0.0700 & 0.0705 & 0.0718 & 3.0341 & 0.7929 & 0.4485 & 0.2005 & 1.0935 & 0.6364 & 0.3223 & 0.1539 \\
    & 10 & 0.0698 & 0.0747 & 0.0818 & 0.0961 & 2.5788 & 0.9237 & 0.5221 & 0.2941 & 2.1868 & 0.7882 & 0.4132 & 0.2198 \\
100 & 1  & 0.0689 & 0.0693 & 0.0699 & 0.0698 & 3.6654 & 0.8278 & 0.4626 & 0.2154 & 1.3818 & 0.7264 & 0.3452 & 0.1644 \\
    & 5  & 0.0686 & 0.0692 & 0.0707 & 0.0722 & 2.9162 & 0.9129 & 0.5034 & 0.2500 & 2.1448 & 0.7755 & 0.3853 & 0.1885 \\
    & 10 & 0.0709 & 0.0757 & 0.0830 & 0.0986 & 2.7325 & 0.9802 & 0.5942 & 0.3348 & 1.6324 & 1.0800 & 0.5813 & 0.2699 \\
\hline
\end{tabular}}
\label{tab:VARMA_bias_RMSE_1000}
\end{table}

\end{document}

%% file: preamble.ltx
\usepackage[margin=1in]{geometry}
\usepackage{parskip}
\usepackage{setspace}
\usepackage{mathptmx} 
\usepackage[scaled=0.92]{helvet} 
\usepackage{courier} 

\usepackage{sectsty}
\sectionfont{\sffamily\Large\bfseries}
\subsectionfont{\sffamily\large\bfseries}
\subsubsectionfont{\sffamily\normalsize\bfseries}

\usepackage[utf8]{inputenc}  
\usepackage[T1]{fontenc}     
\usepackage[english]{babel}

\usepackage{fancyhdr}
\usepackage{amsmath, amssymb, amsthm}
\usepackage{mathtools}
\usepackage{mathrsfs}
\usepackage{bm}

\newtheorem{theorem}{Theorem}
\newtheorem{assumption}{Assumption}
\newtheorem{lemma}{Lemma}
\newtheorem{proposition}{Proposition}
\newtheorem{corollary}{Corollary}
\newtheorem{definition}{Definition}

\usepackage{graphicx}
\usepackage{float}
\usepackage{caption}
\usepackage{subcaption}
\usepackage{booktabs}
\usepackage{multirow}

\usepackage{tikz}
\usepackage{pgfplots}
\pgfplotsset{compat=1.18}

\usepackage[authoryear]{natbib}
\usepackage[colorlinks=true, linkcolor=blue, citecolor=blue, urlcolor=blue]{hyperref}

\usepackage{algorithm}
\usepackage{algpseudocode}

\usepackage{orcidlink}
\usepackage{xcolor}

\usepackage{enumitem}
\usepackage{comment}